\documentclass[runningheads,a4paper]{llncs}

\usepackage{amssymb}
\usepackage{graphicx}
\usepackage{algpseudocode}

\usepackage{url}
\urldef{\mailsa}\path|{sarelcoh, fiat, haimk}@post.tau.ac.il|
\urldef{\mailsb}\path|moshik1@gmail.com|
\newcommand{\keywords}[1]{\par\addvspace\baselineskip
\noindent\keywordname\enspace\ignorespaces#1}

\def\bs{\mbox{\rm bs}}
\def\bkey{\mbox{\rm bkey}}

\newcommand{\OR}{{\mathrm{OR}}}
\newcommand{\AND}{{\mathrm{AND}}}

\begin{document}

\mainmatter  

\title{Minimal Indices for Successor Search}
\subtitle{[Full Version]}

\titlerunning{Minimal Indices for Successor Search}

%
%
\author{Sarel Cohen\and Amos Fiat\and Moshik Hershcovitch\and Haim Kaplan}
\authorrunning{Minimal Indices for Successor Search}

\institute{School of Computer Science, Tel Aviv University
\mailsa\\
\mailsb\\}

%
%

\toctitle{Minimal Indices for Successor Search}
\tocauthor{Extended Abstract}
\maketitle

\begin{abstract}
We give a new successor data structure which improves upon the index
size of the P\v{a}tra\c{s}cu-Thorup data structures, reducing the
index size from $O(n w^{4/5})$ bits to $O(n \log w)$ bits, with
optimal probe complexity. Alternatively, our new data structure can be
viewed as matching the space complexity of the (probe-suboptimal)
$z$-fast trie of Belazzougui et al. Thus, we get the best of both
approaches with respect to both probe count and index size. The
penalty we pay is an extra $O(\log w)$ inter-register operations.
Our data structure can also be used to solve the weak prefix search
problem, the index size of $O(n \log w)$ bits is known to be optimal
for any such data structure.

The technical contributions include highly efficient single word
indices, with out-degree $w/\log w$ (compared to the $w^{1/5}$
out-degree of fusion tree based indices). To construct such high
efficiency single word indices we device highly efficient bit
selectors which, we believe, are of independent interest.
\end{abstract}

\section{Introduction}

\label{sec:intro}

A fundamental problem in data structures is the successor problem:
given a RAM with $w$ bit word operations, and $n$ keys (each $w$
bits long), give a data structure that answers successor queries
efficiently. We distinguish between the space occupied by the $n$
input keys themselves, which is $O(nw)$ bits, and the additional
space requires by the data structure which we call the {\em index}.
The two other performance measures of the data structure which are
of main interest are how many accesses to memory (called {\em
probes}) it performs per query,  and  the query time or the total
number of machine operations performed per query, which could be
larger than the number of probes. We can further distinguish between
probes to the index and probes to the input keys themselves. The
motivation is that if the index is small and fits in cache probes to
the index would be cheaper. We focus  on constructing a data
structure for the successor problem that requires sublinear $o(nw)$
extra bits.


The simplest successor data structure is a sorted list, this
requires no index, and performs  $O(\log n)$ probes and $O(\log n)$
operations per binary search. This high number of probes that are
widely dispersed can makes this solution inefficient for large data
sets.

Fusion trees of Fredman and Willard \cite{FredmanWillard} (see also
\cite{FredmanWillard2}) reduce the number of probes and time to
$O(\log_w n)$. A fusion tree node
 has outdegree $B=w^{1/5}$ and therefore fusion trees require
only  $O(nw/B) = O(n w^{4/5})$ extra bits.

Another famous data structure is the $y$-fast trie of Willard
\cite{Dan:1983}. It requires linear space ($O(nw)$ extra bits) and
$O(\log w)$ probes and time per query.

P\v{a}tra\c{s}cu and Thorup \cite{Patrascu:lower1} solve the
successor problem optimally (to within an $O(1)$ factor) for any
possible point along the probe count/space tradeoff, and for any
value of $n$ and $w$. However, they do not distinguish between the
space required to store the input and the extra space required for
the index. They consider only the total space which cannot be
sublinear.

P\v{a}tra\c{s}cu and Thorup's linear space data structure  for
successor search is an improvement of three previous data-structures
and achieves the following bounds.

\vspace*{-0.2cm}
\begin{enumerate}
\item
For values of $n$ such that $\log n \in [0, \frac{\log^2 w}{\log
\log w}]$ their data structure is a fusion tree and therefore the
query time is $O(\log_w n)$. This bound increases monotonically with
$n$.
\item For $n$ such that $\log n \in [\frac{\log^2 w}{\log \log w},
\sqrt{w}]$ their data structure is a generalization of the data
structure of Beame \& Fich \cite{Beame01optimalbounds} that is suitable for linear
space, and has the bound $O(\frac{\log w}{\log \log w - \log \log
\log n})$. This bound increases from $O(\frac{\log w}{ \log\log w})$
at the beginning of this range to $O(\log w)$ at the end of the
range.
\item For values of $n$ such that $\log n \in [\sqrt{w},w]$
their data structure is a slight improvement of the van Emde Boas'
(vEB) data structure \cite{DBLP:journals/ipl/Boas77} and has the bound of $O(\max \{ 1 ,\log
(\frac {w -\log n}{\log w} ) \})$. This bound decreases with $n$
from $O(\log w)$ to $O(1)$.
\end{enumerate}

\looseness=-1 A recent data structure of Belazzougui et
al.~\cite{Belazzougui:2009} called  the probabilistic $z$-fast trie,
 reduces the extra space requirement to $O(n \log w)$ bits, but requires a (suboptimal)
expected $O(\log w)$ probes (and $O(\log n)$ probes in the worst
case). See Table \ref{tab-compareoldnew} for a detailed comparison
between various data structures for the successor porblem with
respect to the space and probe parameters under consideration.

\begin{table}[htb]
\centering \scalebox{0.85}{
 {\begin{tabular}{|l|l|l|c|l|l|l|}
  \hline
  Data Structure & Ref. &
  \begin{tabular}{l}
    Index
    size \\
    (in bits)
  \end{tabular}  & &
    \begin{tabular}{l}
    $\#$ Non-index \\
    Probes
  \end{tabular}&
     \begin{tabular}{l}
    Total \\
    $\#$ Probes
  \end{tabular} &
    $\#$ operations
  \\
  \hline
  Binary Search & & -- & & $O(\log n)$ & $O(\log n)$ & $O(\#
  \mathrm{probes})$ \\
\hline
    van Emde Boas  & \cite{DBLP:journals/ipl/Boas77} & $O(2^w)$  & &  $O(1)$&
  $O(\log w)$ & $O(\# \mathrm{probes})$ \\ \hline
  $x$-fast trie & \cite{Dan:1983} & $O(nw^2)$ & &  $O(1)$& $O(\log w)$
  & $O(\# \mathrm{probes})$\\ \hline
  $y$-fast trie & \cite{Dan:1983} & $O(nw)$ & &   $O(1)$& $O(\log w)$
  & $O(\# \mathrm{probes})$\\
  \hline
  \begin{tabular}{c}
    $x$-fast trie \\
    on ``splitters" \\
    poly($w$) apart
  \end{tabular}  & \begin{tabular}{c}
    Folklore
  \end{tabular} & $O(n/\mathrm{poly}(w))$  & &  $O(\log w)$&
  $O(\log w)$ & $O(\# \mathrm{probes})$ \\
  \hline
    Beame $\&$ Fich & \cite{Beame01optimalbounds} & $\Theta(n^{1+\epsilon} w)$ & &  $O(1)$& $
  O\left(\frac{\log w}{\log \log w}\right)$ & $O(\# \mathrm{probes})$ \\
  \hline
  Fusion Trees & \cite{FredmanWillard} & $O(n w^{4/5})$ & &  $O(1)$& $O\left(\frac{\log n}{\log w}\right)$
  & $O(\# \mathrm{probes})$\\    \hline
  $z$-fast trie & \cite{Belazzougui:2009,Belazzougui:2008,Belazzougui10}& $O(n \log w)$ &
  \begin{tabular}{c}
    exp.  \\ \hline
    w.c.
      \end{tabular} &
\begin{tabular}{l}
    $O(1)$ \\ \hline
   $O(\log n)$
      \end{tabular} &
  \begin{tabular}{l}
     $O(\log w)$ \\ \hline
     $O(\log n)$
      \end{tabular}  & $O(\# \mathrm{probes})$  \\
      \hline
    P\v{a}tra\c{s}cu $\&$ Thorup & \cite{Patrascu:lower1} & \begin{tabular}{l}
    $O(nw)$ or\\
    $O(nw^{4/5})$
      \end{tabular}  & &  $O(1)$&  \begin{tabular}{l}
     Optimal given\\
      linear space \\
      \end{tabular} & $O(\# \mathrm{probes})$  \\
      \hline
     \begin{tabular}{l}
     P\v{a}tra\c{s}cu $\&$ Thorup \\
     + $\gamma$-nodes
      \end{tabular} & This Paper & $O(n \log w)$  & &  $O(1)$ & \begin{tabular}{l}
     Optimal given\\
      linear space \\
      \end{tabular} & \begin{tabular}{l}
     $O(\# \mathrm{probes}$\\
     $ + \log w)$
      \end{tabular}
         \\
\hline
   \end{tabular}}}

\caption{Requirements of various data structures for the successor
problem. The word length is $w$ and the number of keys is $n$.
Indexing groups of $w/\log w$ consecutive keys with our new word
indices we can reduce the space of any of the linear space data
structures above to $O(n\log w)$ bits while keeping the number of
probes the same and increasing the query time by $O(\log w)$.}
\label{tab-compareoldnew}
 \vspace*{-0.7cm}
\end{table}

Consider the following multilevel scheme to reduce index size:  (a)
partition the keys into consecutive sets of $w^{1/5}$ keys, (b)
build a Fusion tree index structure for each such set (one $w$ bit
word), and (c) index the smallest key in every such group  using any
linear space data structure. The number of fusion tree nodes that we
need $n/w^{1/5}$ and the total space required for these nodes and
the data structure that is indexing them is $O(n w^{4/5})$.

This standard bucketing trick shows that we can get indices of
smaller size by constructing  a ``fusion tree node'' of larger
outdegree. That is we seek a data structure, which we refer to as a
{\em word-index}, that by using $O(1)$ words can answer successor
queries with respect to as many keys as possible.

Our main contribution is such a word index that can handle $w/\log
w$ keys (rather than $w^{1/5}$ for fusion trees).\footnote{The
$w/\log w$ keys take more than $O(1)$ words but are not considered
part of the word index.} However, this new highly compact index
requires $\Theta(\log w)$ operations per search (versus the $O(1)$
operations required by Fusion trees).

Using these word indices we obtain, as described  above, a
(deterministic) data structure that, for any $n$, $w$, answers a
successor query with an optimal number of probes (within an $O(1)$
factor), and requires only $O(n \log w)$ extra bits. We remark that
we only probe $O(1)$ non-index words (which is true of
P\v{a}tra\c{s}cu-Thorup data structures as well, with minor
modifications). The penalty we pay is an additional $O(\log w)$ in
the time complexity.

Indices of small size are particularly motivated today by the
multicore shared memory architectures abundant
today~\cite{Drepper07whatevery,Shavit}. When multiple cores access
shared cache/memory, contention arises. Such contention is deviously
problematic because it may cause serialization of memory accesses,
making a mockery of multicore parallelism. Multiple memory banks and
other hardware are attempts to deal with such problems, to various
degrees. Thus, the goals of reducing the index size, so it fits
better in fast caches, reducing the number of probes extraneous to
the index, and the number of probes within the index, become
critical.
%
\section{High level overview of our results and their implications}
\noindent {\bf Computation model:} We assume
 a RAM model of computation with $w$ bits per word.
A key (or query) is one word ($w$ bits long).
 We can operate on the registers using at least a basic instruction set consisting of (as defined in \cite{DBLP:conf/wads/BrodnikMM97}): Direct and
indirect addressing, conditional jump, and a number of {\sl
inter-register operations}, including addition, subtraction, bitwise
Boolean operations and left and right shifts. All operations are
unit cost. One of our construction does not require multiplication.

\smallskip

 We give three
variants of high outdegree single word indices which we call
$\alpha$ nodes, $\beta$ nodes, and $\gamma$ nodes. Each of these
structures index $w/\log w$ keys and answer successor queries using
only $O(1)$ $w$-bit words, $O(\log w)$ time, and $O(1)$ extra-index
probes (in expectation for $\alpha$ and $\beta$ nodes, worst case
for $\gamma$ nodes) to get at most two of the $w/\log w$ keys.

\looseness=-1 The $\alpha$ node is simply a $z$-fast trie (\cite{Belazzougui:2009}) applied to $w/ \log w$
keys. Given the small number of keys, the $z$-fast trie can be
simplified. A major component of the $z$-fast trie involves
translating a prefix match (between a query and a trie node) to the
query rank. As there are only $w/\log w$ keys involved, we can
discard this part of the $z$-fast trie and store ranks explicitly in
$O(1)$ words.

Based on a different set of ideas, $\beta$ nodes are arguably
simpler than the $z$-fast trie, and have the same performance as the
$\alpha$ nodes. As $\beta$-nodes are not our penultimate construction, the full description of $\beta$-nodes is in Appendix \ref{sec:beta-full}.

Our penultimate variant, $\gamma$ nodes, has the advantages that it is
deterministic and gives worst case $O(1)$ non-index probes, and,
furthermore, requires no multiplication.

To get the $\gamma$ nodes we introduce highly efficient
bit-selectors (see section \ref{subsec:bitextract}) that may be of
independent interest. Essentially, a bit-selector selects a multiset
of bits from a binary input string and outputs a rearrangement of
these bits within a shorter output string.

 Thorup
\cite{thorup03}
 proved that it is impossible to have
$O(1)$ time successor search  in a ``standard AC(0) model", for any
non-constant number of keys, unless one uses enormous space,
$2^{\Omega(w)}$, where $w$ is the number of bits per word. This
means that it would be impossible to derive an improved
$\gamma$-node (or Fusion tree node) with $O(1)$ time successor
search in the ``standard AC(0) model''.

\subsection{Succinct successor data structure}

As mentioned in the introduction  we obtain using our word indices a
successor data structure that requires $O(n\log w)$ bits in addition
to the input keys. The idea is standard and simple: We divide the
keys into consecutive chunks of size $w/\log w$ keys each. We index
each chunk with one of our word indices and index the chunks (that
is the first key in each chunk) using another linear space data
structure. This has the following consequences depending upon the
linear space data structure which we use to index the chunks. (We
henceforth refer to our $\gamma$-nodes, but similar results can be
obtained using either $\alpha$ or $\beta$ nodes in expectation.)

\textbf{Fusion Trees + $\gamma$-nodes:} This data structure answers
successor queries with $O(\log_w n)$ probes, and  $O(\log_w n + \log
w)$ time.

\textbf{The optimal structure of P\v{a}tra\c{s}cu $\&$ Thorup  +
$\gamma$-nodes:} Here the  number of probes to answer a query is
optimal, the time is $O(\# probes + \log w)$.

\textbf{$y$-fast-trie + $\gamma$-nodes:} This gives an improvement
upon the recently introduced [probabilistic] $z$-fast-trie,
\cite{Belazzougui:2009,Belazzougui:2008} (we omit the
``probabilistic" prefix hereinafter). The worst-case probes and query time
improves from $O(\log n)$ to $O(1)$ probes and $O(\log w)$ query time, and the data structure
is deterministic.

\noindent {\bf The weak prefix search problem:} In this problem the
query is as follows. Given a  bit-string $p$, such that $p$ is the
prefix of at least one key among the $n$ input keys, return the
range of ranks of those  input keys having $p$ as a prefix.

It is easy to modify the index of our successor data structures to
 a new data structure for ``weak prefix search". We construct
a word $x$ containing the query
 $p$ padded to the right with trailing zeros, and a word
 $y$ containing the query
 $p$ padded to the right with trailing ones.
 Searching for the rank of the successor of $x$ in $S$
and the rank of the predecessor of $y$ in $S$ gives the required
range.

We note that we can carry out the search of the successor of $x$ and
the predecessor of $y$ without accessing the keys indexed by the
$\gamma$ nodes. As we will see, our $\gamma$ nodes implement a
succinct  blind tree. Searching a blind trie for the right rank of
the successor typically requires accessing one of the indexed keys.
But, as implicitly used in \cite{Belazzougui10}, this access can be
avoided if the query is a padded prefix of an indexed key such as
$x$ and $y$ above. This implies that the keys indexed by the
$\gamma$ nodes can in fact be discarded and not stored at all. We
get a data structure of overall size $O(n\log w)$ bits for weak
prefix search.

Belazzougui et {\sl al.}, \cite{Belazzougui10}, show that any data
structure supporting ``weak prefix search'' must have size $\Omega(n
\log w)$ bits. Hence, our index size is optimal for this related
problem.

\subsection{Introducing Bit-Selectors and Building a $(k,k)$-Bit Selector} \label{subsec:bitextract}

 To construct the $\gamma$-nodes we define and construct bit selectors as follows.
 A $(k,L)$ bit-selector, $1 \leq k \leq L \leq w$, consists
of a preprocessing phase and a query phase, (see Figure
\ref{fig-extractor1}):

\begin{itemize}
\item  The preprocessing phase: The input is a sequence of length $k$ (with repetitions),
$$I= I[1], I[2], \ldots, I[k],$$
where $0\leq I[j] \leq w-1$ for all
$j=1,\ldots,k$. Given $I$, we compute the following: \begin{itemize}
\item A sequence of $k$ strictly increasing indices, \\$0 \leq j_1
< j_2 < \cdots < j_k \leq L-1$, and,
\item An $O(1)$ word data structure, $D(I)$.
\end{itemize}
\item The query phase: given an input word $x$, and using $D(I)$, produces an
output word $y$ such that
\begin{eqnarray*} y_{j_\ell} &=& x_{I[\ell]}, \qquad 1 \leq \ell \leq k, \\
y_m &=&0, \qquad m\in \{0,\ldots w-1\} - \{j_\ell\}_{\ell=1}^k.
\end{eqnarray*}
\end{itemize}

One main technically difficult result is a construct for $(k, k)$
bit-selectors for all $1 \leq k \leq w/\log w$ (Section
\ref{sec:extract}). The bit selector query time is $O(\log w)$,
while the probe complexity and space are constant.

With respect to upper bounds, Brodnik, Miltersen, and Munro
\cite{DBLP:conf/wads/BrodnikMM97}, give several bit manipulation
primitives, similar to some of the components we use for
bit-selection, but the setting  of \cite{DBLP:conf/wads/BrodnikMM97}
is somewhat different, and there is no attempt to optimize criteria
such as memory probes and index size. The use of Benes networks to
represent permutations also appears in Munro et. al
\cite{DBLP:conf/icalp/MunroRRR03}.

Note that, for $(k,k)$-bit-selectors, it must be that $j_\ell =
\ell-1$, $1 \leq \ell \leq k$, independently of $I$. For a sequence
of indices $I$, we define $x[I]$ to be the bits of $x$ in these
positions (ordered as in $I$), if $I$ has multiplicities then $x[I]$
also has multiplicities. With this notation a $(k,k)$ bit selector
$D(I)$ computes $x[I]$ for a query $x$ in $O(\log w)$ time.

\looseness=-1 A $(w^{1/5},w^{4/5})$ bit-selector is implicit in fusion trees and lie
at the core of the data structure. Figure \ref{fig-compareext}
compares the fusion tree bit-selector with our construction.

\begin{figure*}[ht]
\centering {
\begin{tabular}{|l|l|c|c|c|c|}
  \hline
    & $k$ & $L$ & $|D(I)|$ in words & $\#$ Operations & Multiplication?\\
  \hline
  \begin{tabular}{c}
    Fusion tree \\
    bit-selector
  \end{tabular}  & $1\leq k\leq w^{1/5}$ & $k^4$ & $O(1)$ & $O(1)$ & Yes \\
  \hline
  Our bit-selector & $1 \leq k\leq w/\log w$ & $k$ & $O(1)$ & $O(\log w)$ & No \\
  \hline
\end{tabular}}
\caption{The bit-selector used for Fusion Trees in
\cite{FredmanWillard,FredmanWillard2} {\sl vs.} our bit-selector.}
\label{fig-compareext}
\end{figure*}

We remark that  Andersson, Miltersen, and Thorup
\cite{Andersson:ac0} give an AC(0) implementation of fusion trees,
{\sl i.e.}, they use special purpose hardware to implement a $(k,k)$
bit-selector (that produces a sketch of length $k$ containing $k$
bits of the key). Ignoring other difficulties, computing a [perfect]
sketch in AC(0) is easy: just lead wires connecting the source bits
to the target bits. With this interpretation, our bit-selector is a
software implementation in $O(\log w)$ time that implements the
special purpose hardware implementation of \cite{Andersson:ac0}.

Our bit-selectors are  optimal with respect to query time, when
considering implementation on a ``practical RAM'' (no multiplication
is allowed) as defined by
 Miltersen \cite{Miltersen96}.
 This follows from Brodnik
et {\sl al.} \cite{DBLP:conf/wads/BrodnikMM97} (Theorem 17) who
prove that in the ``practical  RAM'' model, any
$(k,k)$-bit-selector, with $k\geq \log^{10} w$, requires at least
$\Omega(\log k)$ time per bit-selector query. (Observe that the
bit-reversal of Theorem 17 in \cite{DBLP:conf/wads/BrodnikMM97} is a
special case of bit-selection).

\section{Bit Selectors}
\label{sec:extract}

In this section we describe both the preprocessing and selection
operations for our bit-selectors. We sketch the selection process,
which makes use of $D(I)$, the output of the preprocessing. A more
extensive description and figures can be found in the appendix,
Section \ref{sec:extract-full}.

$D(I)$ consists of $O(1)$ words and includes precomputed constants
used during the selection process. As $D(I)$ is $O(1)$ words, we
assume that $D(I)$ is loaded into registers at the start of the
selection process. Also, the total working memory required
throughout the selection is $O(1)$ words, all of whom we assume to
reside within the set of registers.

Partition the sequence $\sigma= 0,1, \ldots,w-1$ into $w/\log w $
blocks (consecutive, disjoint, subsequences of $\sigma$), each of
length $\log w$. Let $B_j$ denote the $j$th block of a word, {\sl
i.e.}, $B_j = j\log w, j\log w + 1, \ldots, (j+1)\log w -1 $,  $0
\le j \le w/\log w -1$.

Given an input word $x$ and the precomputed $D(I)$, the selection
process goes through the seven phases sketched below.

In this high level explanation we give an example input using the following parameters:
The word length $w=16$ bits, a bit index requires $\log w = 4$ bits, $I$ consists of $w/\log w = 4$ indices (with repetitions). A ``block" consists of $\log w=4$ bits, and there are $w/\log w = 4$ blocks.

As a running example let the input word be $x={\bf 1}  0  0  0 \
1 1  0  1 \    1  1  1  0 \    {\bf 0}  0  1  {\bf 1}$ and let
$I=<0,15, 12, 15>$, the required output is $x[I]=1101$.

\noindent{\bf Phase 0:}
Zero irrelevant bits.  We take the mask $M$ with ones at positions in $I$, and set $x=x \ \AND\ M$. For our example this gives
\\Input \ \ : $M={\bf 1}  0  0  0  \   0  0  0  0 \    0  0  0  0  \   {\bf 1}  0  0  {\bf 1}$, $x={\bf 1}  0  0  0 \    1  1  0  1 \    1  1  1  0 \    {\bf 0}  0  1  {\bf 1}$;
\\Phase 0: $M={\bf 1}  0  0  0  \   0  0  0  0 \    0  0  0  0  \   {\bf 1}  0  0  {\bf 1}$, $x={\bf 1}  0  0  0\     0  0  0  0\     0  0  0  0 \    {\bf 0}  0  0  {\bf 1}$.
\par
\noindent{\bf Phase 1:}
  \looseness=-1 Packing blocks to the Left: All bits of $x$ whose
  index belongs to some block are shifted to the left within the
  block. We modify the mask $M$ accordingly. Let the number of such bits in block $j$ be $b_j$.
  This phase transforms $M$ and $x$ as follows:
\\Phase 0: $M={\bf 1}  0  0  0  \   0  0  0  0 \    0  0  0  0  \   {\bf 1}  0  0  {\bf 1}$, $x={\bf 1}  0  0  0\     0  0  0  0\     0  0  0  0 \    {\bf 0}  0  0  {\bf 1}$;
\\Phase 1: $M={\bf 1}  0  0  0 \    0  0  0  0  \   0  0  0  0 \    {\bf 1}  {\bf 1}  0  0$, $x = {\bf 1}  0  0  0 \    0  0  0  0  \   0  0  0  0 \    {\bf 0}  {\bf 1}  0  0$;
\\Note that $b_0=1$, $b_1=b_2=0$, and $b_3=2$. Phase 1 requires $O(\log w)$  operations on a constant number of words (or registers).
  \par
\noindent{\bf Phase 2:}
  Sorting Blocks  in descending order of $b_j$
  (defined in Phase 1 above).
This phase transforms $M$ and $x$ as follows:
\\   Phase 1: $M={\bf 1}  0  0  0 \    0  0  0  0  \   0  0  0  0 \    {\bf 1}  {\bf 1}  0  0$, $x = {\bf 1}  0  0  0 \    0  0  0  0  \   0  0  0  0 \   {\bf 0}  {\bf 1}  0  0$;
\\   Phase 2: $M={\bf 1}  {\bf 1}  0  0  \   {\bf 1}  0  0  0  \   0  0  0  0  \   0  0  0  0$, $x =  {\bf 0}  {\bf 1}  0  0  \   {\bf 1}  0  0  0  \   0  0  0  0   \  0  0  0  0$;
\\
   Technically, phase 2 uses a Benes network to sort the blocks in descending order of $b_j$, in our running example this means block 3 should come first, then block 0, then blocks 2 and 3 in arbitrary order. Brodnik, Miltersen, and Munro \cite{DBLP:conf/wads/BrodnikMM97} show how to simulate a Benes network on bits of a word, we extend this so as to sort entire blocks of $\log w$ bits.
   \\The precomputed $D(I)$ includes $O(1)$ words to encode this Benes network. Phase 2 requires $O(\log w)$ bit operations on $O(1)$ words.
   \par
\noindent{\bf Phase 3:}
  Dispersing bits:
   reorganize the word produced in Phase 2 so that each
  of the different bits whose index is in $I$ will occupy the
  leftmost
  bit of a unique block. As there may be less distinct indices in
  $I$ than blocks, some of the blocks may be empty, and these will
  be the rightmost blocks. This process requires $O(\log w)$ word
  operations to reposition the bits.
This phase transforms $M$ and $x$ as follows:
\\  Phase 2: $M={\bf 1}  {\bf 1}  0  0  \   {\bf 1}  0  0  0  \   0  0  0  0  \   0  0  0  0$, $x =  {\bf 0}  {\bf 1}  0  0  \   {\bf 1}  0  0  0  \   0  0  0  0   \  0  0  0  0$;
\\  Phase 3: $M={\bf 1}  0  0  0\     {\bf 1}  0  0  0  \   {\bf 1}  0  0  0 \    0  0  0  0$, $x =  {\bf 0}  0  0  0  \   {\bf 1}  0  0  0  \   {\bf 1}  0  0  0 \    0  0  0  0$;
 \par
\noindent{\bf Phase 4:} Packing bits. The
  goal now is to move the bits positioned by Phase 3 at the leftmost
  bits of the leftmost $r$ blocks ($r$ being the number of indices in $I$ without repetitions). Again, by appropriate bit
  manipulation, this can be done with $O(\log w)$ word operations (see appendix).
  This phase transforms $M$ and $x$ as follows:
\\  Phase 3: $M={\bf 1}  0  0  0\     {\bf 1}  0  0  0  \   {\bf 1}  0  0  0 \    0  0  0  0$, $x =  {\bf 0}  0  0  0  \   {\bf 1}  0  0  0  \   {\bf 1}  0  0  0 \    0  0  0  0$;
\\Phase 4: $M={\bf 1}  {\bf 1}  {\bf 1}  0  \   0  0  0  0  \   0  0  0  0  \   0  0  0  0$, $x={\bf 0}  {\bf 1}  {\bf 1}  0  \   0  0  0  0  \   0  0  0  0 \    0  0  0
0$;\\
\looseness=-1 We remark that if $r=k$, {\sl i.e.}, if $I$
  contains no duplicate indices, then we can skip Phases 5 and 6
  whose purpose is to duplicate those bits required several times in
  $I$.
 \par
\noindent{\bf Phase 5:}
  Spacing the bits.
  Once again, we simulate a Benes network on the $k$ leftmost bits.
  The purpose of this permutation is to space out and rearrange the
  bits so that bits who appear multiple times in $I$ are placed so
  that multiple copies can be made. \\\looseness=-1 In our running example, phase 5 changes neither $M$ nor $x$, but this is coincidental -- for other inputs ($I' \neq I$) phase 5 would not be the identity function. Phase 5 is yet another application of a Benes network and requires $O(\log w)$ word operations. \par
\noindent{\bf Phase 6:}
   Duplicating bits - we duplicate the bits for which space was prepared
  during Phase 5.
    This phase transforms $M$ and $x$ as follows:
\\Phase 5: $M={\bf 1}  {\bf 1}  {\bf 1}  0  \   0  0  0  0  \   0  0  0  0  \   0  0  0  0$, $x={\bf 0}  {\bf 1}  {\bf 1}  0  \   0  0  0  0  \   0  0  0  0 \    0  0  0  0$;
\\Phase 6: $M= {\bf 1}  {\bf 1}  {\bf 1}  {\bf 1} \    0  0  0  0  \   0  0  0  0  \   0  0  0  0$, $x={\bf 0}  {\bf 1}  {\bf 1}  {\bf 1}  \   0  0  0  0   \  0  0  0  0  \   0  0  0  0$;
\\
 Technically, phase 6 makes use of shift and $\OR$ operations, where the shifts are decreasing powers of two.
\par
\noindent{\bf Phase 7:}
  \looseness=-1 Final positioning: The bits are all now
  in the $k$ leftmost positions of a word, every bit
  appears the same number of times it's index appears in $I$, and we
  need to run one last Benes network simulation so as to permute
  these $k$ bits. This permutation gives the final outcome.
  This phase transforms $M$ and $x$ as follows:
\\Phase 6: $M= {\bf 1}  {\bf 1}  {\bf 1}  {\bf 1} \    0  0  0  0  \   0  0  0  0  \   0  0  0  0$, $x={\bf 0}  {\bf 1}  {\bf 1}  {\bf 1}  \   0  0  0  0   \  0  0  0  0  \   0  0  0  0$;
\\Phase 7: $M= {\bf 1}  {\bf 1}  {\bf 1}  {\bf 1} \    0  0  0  0  \   0  0  0  0  \   0  0  0  0$, $x={\bf 1}  {\bf 1}  {\bf 0}  {\bf 1}  \   0  0  0  0  \   0  0  0  0   \  0  0  0  0$;
\\
  \looseness=-1 Note the leftmost $|I|=w/\log w = 4$ bits of $x$ contain the required output of the bit selector.\par

\section{$\gamma$-nodes}
\label{sec-gamma}
 In this section we use the $(w/\log w, w/\log
w)-$bit-selector, described above, to build a {\em $\gamma$-node}
defined as follows.

\begin{definition}
A $\gamma$-node answers successor queries over a static set $S$ of
at most  $w/\log w$  $w$-bit keys. The $\gamma$-node uses a compact
index of $O(1)$ $w$-bit words, in addition to the input $S$.
Successor queries perform $O(1)$ word probes, and $O(\log w)$
operations.
\end{definition}

We describe the $\gamma$-node data structure in stages, beginning
with a {\em slow $\gamma$-node} below. A slow $\gamma$-node is
defined as a $\gamma$-node but performs $O(w/\log w)$ operations
rather than $O(\log w)$.

\subsection{Construction of Slow $\gamma$-nodes}

We build a blind trie over the set of keys $S= y_1 < y_2,\ldots,<
y_k$, $k \le w/\log w$. We denote this trie by $T(S)$.
  The trie $T(S)$ is a full binary tree with $k$ leaves, each corresponds
to a key, and $k-1$ internal nodes. (We do not think of the keys as
part of the trie.) We store $T(S)$ in $O(1)$ $w$-bit words. (The
keys, of course require $|S|w$ bits.) $T(S)$ has the following
structure:

\begin{enumerate} \item Each internal node of $T(S)$ has pointers
to its left and right children.
\item An internal node $u$ includes
a bit index, $i_u$, in the range $0,\ldots,w-1$, $i_u$ is the length
of the longest common prefix of the keys associated with the leaves
in the subtree rooted at $u$.
 \item
Key $y_i$ corresponds to the $i$th leaf from left to right. We store
$i$ in this leaf and denote this leaf by $\ell(y_i)$.
\item
Keys associated with descendants of the left-child of $u$ have bit
$i_u$ equals to zero. Analogously, keys associated with descendants
of the right-child of $u$ have bit $i_u$ equals to one.
\end{enumerate}

In addition to $T(S)$, we assume that the keys in $S$ are stored in
memory, consecutively in sorted order.

Indices both in internal nodes and leaves are in the range
$0,\ldots, w-1$ and thereby require $O(\log w)$ bits. Since $T(S)$
has $O(w/\log w)$ nodes, a pointer to a node also
 requires $O(\log w)$ bits.
Thus, in total, each  node in $T(S)$ requires only $O(\log w)$ bits.
It follows that $T(S)$ (internal nodes and leaves) requires only
$O(w)$ bits (or, equivalently, can be packed into $O(1)$ words).

Fundamentally, a \emph{blind-search} follows a root to leaf path in
blind trie $T(S)$, ignoring intermediate bits. Searching $T(S)$ for
a query $x$ always ends at leaf of the trie (which contains the
index of some key). Let $\bs(x,S)$ denote the index stored at this
leaf, and let $\bkey(x)$ be $y_{\bs(x,S)}$. {\sl I.e.}, blind search
for query $x$ in $T(S)$ leads to a leaf that points to $\bkey(x)$.
In general, $\bkey(x)$ is {\sl not} the answer to the successor
query, but it does have the longest common prefix of $x$ amongst all
keys in $S$. (See \cite{Ferragina:1999}.)

To arrive at the successor of $x$, we retrieve $\bkey(x)$ and
compute its longest common prefix with $x$. Let $b$ be the next bit
of $x$, after $\mathrm{LCP}(x,\bkey(x))$. We use $b$ to pad the
remaining bits, let $\|$ denote concatenation, and let
$$z=\mathrm{LCP}(x,\bkey(x)) \| b^{w-|\mathrm{LCP}(x,\bkey(x))|}.$$ Finally, we perform
a second blind-search on $z$. The result of this second search gives
us the index of the successor to $x$ to within $\pm 1$.

Overall, the number of probes required for such a search is $O(1)$.
However, the computation time is equal to the length of the longest
root to leaf path in $T(S)$, which is $O(w/\log w)$.

\subsection{Improving the running time}

 Using our $(w/\log w, w/\log w)$-bit-selector
  we can reduce the search time in the blind trie from $O(w/\log w)$ to $O(\log w)$ operations while still representing
  the trie in $O(1)$  words. For that we change the first part of the query, that is the \emph{blind-search}
  for $\bs(x,S)$  (the index of $\bkey(x)$). Rather than walking top down along a path in the trie
  we use a binary search as follows.

We need the following notation. Any node $u\in T(S)$, internal node
or leaf, defines a unique root to $u$ path in $T(S)$. Denote this
path by $\pi_u = v_0, v_1, \ldots, v_{|\pi_u|}$ where $v_0$ is the
root, $v_{|\pi_u|}=u$, and $v_i$ is the parent of $v_{i+1}$.
For any node $u\in T(S)$ let $I_u$ be the sequence of indices $i_v$
for all internal nodes $v$ along $\pi_u$. Also, let $\zeta_u$ be a
sequence of zeros and ones, one entry per edge in $\pi_u$, zero for
an edge pointing left, one otherwise. For all $1 \leq q \leq |S|$ we
define $\pi_q= \pi_{\ell(y_q)}$, $ I_q = I_{\ell(y_q)}$, and
$\zeta_q = \zeta_{\ell(y_q)}$. The following lemma is
straightforward.

\begin{lemma} \label{lemma:mon}
For any index $1\leq q \leq |S|$, query $x$, we have that
  \begin{eqnarray*}
\zeta_{q} \mbox{\rm\ is lexicographically smaller than\ } x[I_q]
&\Rightarrow&
y_q < \bkey(x)\\
\zeta_q = x[I_q] &\Rightarrow&  y_q=\bkey(x), \\
\zeta_q \mbox{\rm\ is lexicographically larger than\ } x[I_q]
&\Rightarrow&  y_q > \bkey(x).\end{eqnarray*}
\end{lemma}

Based on Lemma \ref{lemma:mon}, given query $x$, we can do binary
search to find $\bs(S,x)$:

\begin{algorithmic}
\State $L\gets 1$, $R\gets |S|$, $q\gets \lfloor (L+R)/2 \rfloor$
\While {$\zeta_{q} \ne x[I_{q}]$}
    \If {$\zeta_{q} < x[I_{q}]$} $R\gets q$
    \Else $\ L\gets q$ \EndIf
    \State $q\gets \lfloor (L+R)/2 \rfloor$
\EndWhile \State \Return $q$
\end{algorithmic}

\begin{lemma}
The above binary search algorithm returns $\bs(x,S)$ and has $O(\log
w)$ iterations.
\end{lemma}

Next we show how to implement each iteration of this binary search
and compare  $x[I_q]$ and $\zeta_q$ in $O(1)$ time while keeping the
trie stored in $O(1)$ words.


For this we devise a sequence $I$ of bit indices, of length $w/\log
w$. Prior to running the binary search  we use the bit-selector of
Section \ref{sec:extract} to compute $x[I]$ and later we use $x[I]$
to construct $x[I_q]$ in every iteration in $O(1)$ time. We extract
$x[I_q]$ from $x[I]$ and retrieve $\zeta_q$ using $O(1)$ additional
words. The details are as follows.

\noindent {\bf The $O(1)$ words which form the $\gamma$ node:} For
each $1 \le q \le |S|$ there is a unique interval $[L_q, R_q]$ of
which $q$ may be the splitting point (i.e.\ $q = \lfloor (L_q+R_q)/2
\rfloor$) during the binary search. Let $\pi_q = u_1, u_2, \ldots,
u_t = \ell(y_q)$ be the path to $q$ as defined above. Define
$j_{L_q} \in 1,\ldots, t$ to be the length of the longest common
prefix of $\pi_q$ and $\pi_{L_q}$. That is
 $u_{j_L}$ is the lowest common ancestor of the leaves
$\ell(y_q)$ and $\ell(y_{L_q})$. Define $j_{R_q}$ analogously, and
let $j=\max(j_{L_q}, j_{R_q})$.

Let $\tilde{\pi}_q$ be the suffix of $\pi_q$ starting at node
$u_{j+1}$, and let $\tilde{I}_{q}$ be the suffix of $I_q$ starting
at $I_q[j+1]$. (These are the indices stored in $u_{j+1}, u_{j+2},
\ldots, u_{t-1}$). Similarly, let $\tilde{\zeta}_{q}$ be the suffix
of $\zeta_{q}$, starting at the $j$th element.

Given $S$, for every $1 \le q \le |S|$ we precompute and store the
following data: $j_{L_q}$, $j_{R_q}$, $\tilde{I_{q}}$,
$\tilde{\zeta_{q}}$. It is easy to verify that  $O(1)$ words suffice
to store the $4 |S|$ values above. Indeed, $j_{L_q}$ and $j_{R_q}$
are indices in $1,\ldots, |S|$, $O(\log w)$ bits each. As the number
of keys $|S| \le w/\log w$,  all the $j_{L_q}$'s, and $j_{R_q}$'s
fit in $O(1)$ words. Since $\tilde{\pi}_{q}$ paths are pairwise
disjoint, the sum of their path lengths is $O(|S|)= O(w/\log w)$.
Hence, storing all the sequences $\tilde{\zeta}_q$, $1\leq q \leq
w/\log w$, requires no more than $O(w/\log w)$ bits. We store the
$\tilde{\zeta}_q$'s concatenated in increasing order of $q$ in a
single word $Z$.

\looseness=-1 The sequence $I$ for which we construct the bit selector is the
concatenation of the $\tilde{I}_q$ sequences, in order of $q$. As
above, it follows that $I$ is a sequence of $O(w /\log w)$ $\log
w$-bit indices. The bit selector $D(I)$ is also stored as part of
the $\gamma$ node.

For each $q$ we also compute and store the index $s_q$ of the
starting position of $\tilde{\zeta}_q$ in $Z$. This is the same as
the index of the starting position of $I_q$ in $I$. Clearly all
these indices $s_q$ can be stored in a single word.

\noindent {\bf Implementing the blind search:} As we mentioned,
given $x$ as a query to the $\gamma$-node, we compute $x[I]$ (once)
from $x$ and $D(I)$, which requires $O(\log w)$ operations and no
more than $O(1)$ probes.

\looseness=-1 At the start of an iteration of the binary search, we have a new
value of $q$, and access to the following values, all of whom are in
$O(1)$ registers from previous iterations:
 $$\quad x[I], \quad j_{L_q}, \quad j_{R_q} \quad L_q, \quad R_q, \quad \zeta_{L_q}, \quad \zeta_{R_q},
\quad x[I_{L_q}], \quad x[I_{R_q}].$$

For the rest of this section let $L=L_q$ and $R=R_q$. We now compute
$\zeta_q$ and $x[I_q]$. We retrieve $j_{L}, j_{R}$ from the
data-structure, and we also retrieve $x[\tilde{I}_{q}]$ from $x[I]$
and $\tilde{\zeta}_{q}$ from $Z$ (note that $x[\tilde{I}_{q}]$ is stored
consecutively in $x[I]$ and $\tilde{\zeta}_{q}$ is stored consecutively in
$Z$, and we use $s_q$ to know where they start).

If $j_L \ge j_R$, we compute $x[I_{q}] \gets (x[I_{L}][1,
\ldots,j_L]) \| (x[\tilde{I_{q}}])$  and $\zeta_{q} \gets
(\zeta_{L}[1,\ldots,j_L]) \| (\tilde{\zeta_{q}})$.

Analogously, if $j_L < j_R$, and we compute
\begin{eqnarray*} x[I_{q}] &\gets& (x[I_{R}][1,\ldots,j_R]) \|
(x[\tilde{I_{q}}])\\ \zeta_{q} &\gets& (\zeta_{R}[1,\ldots,j_R]) \|
(\tilde{\zeta_{q}}).\end{eqnarray*}
\looseness=-1  All these operations are easily
computed using $O(1)$ SHIFT, AND, OR operations.

\section{Open Issues}
\begin{enumerate}
\item Our $(k,k)$-bit selector takes $O(\log w)$ operations,
which are optimal when $k \ge w^{\epsilon}$ for any constant
$\epsilon > 0$. What can be done for smaller values of $k$? ({\sl
E.g.}, for $k=O(1)$ one can definitely do better).

\item It follows from Thorup (\cite{thorup03}) that, in the practical-RAM model, a search node with fan-out $\frac{w}{\log w}$ requires
 $\Omega(\log \log w)$ operations. Our $\gamma$ nodes have fan out $w/\log w$ and require $O(\log w)$ operations. Can this gap be bridged?

\item A natural open question is if the additive $O(\log w)$ in time complexity is required or not.
\end{enumerate}

\subsubsection*{Acknowledgments.} We wish to thank Nir Shavit for introducing us to the problems of
contention in multicore environments, for posing the question of
multicore efficient data structures, and for many useful
discussions. We also wish to thank Mikkel Thorup for his kindness
and useful comments.

\newpage

\appendix



\section{An illustration of a bit-selector}

\begin{figure*}[htbp]
\centering
\includegraphics[scale=2.2]{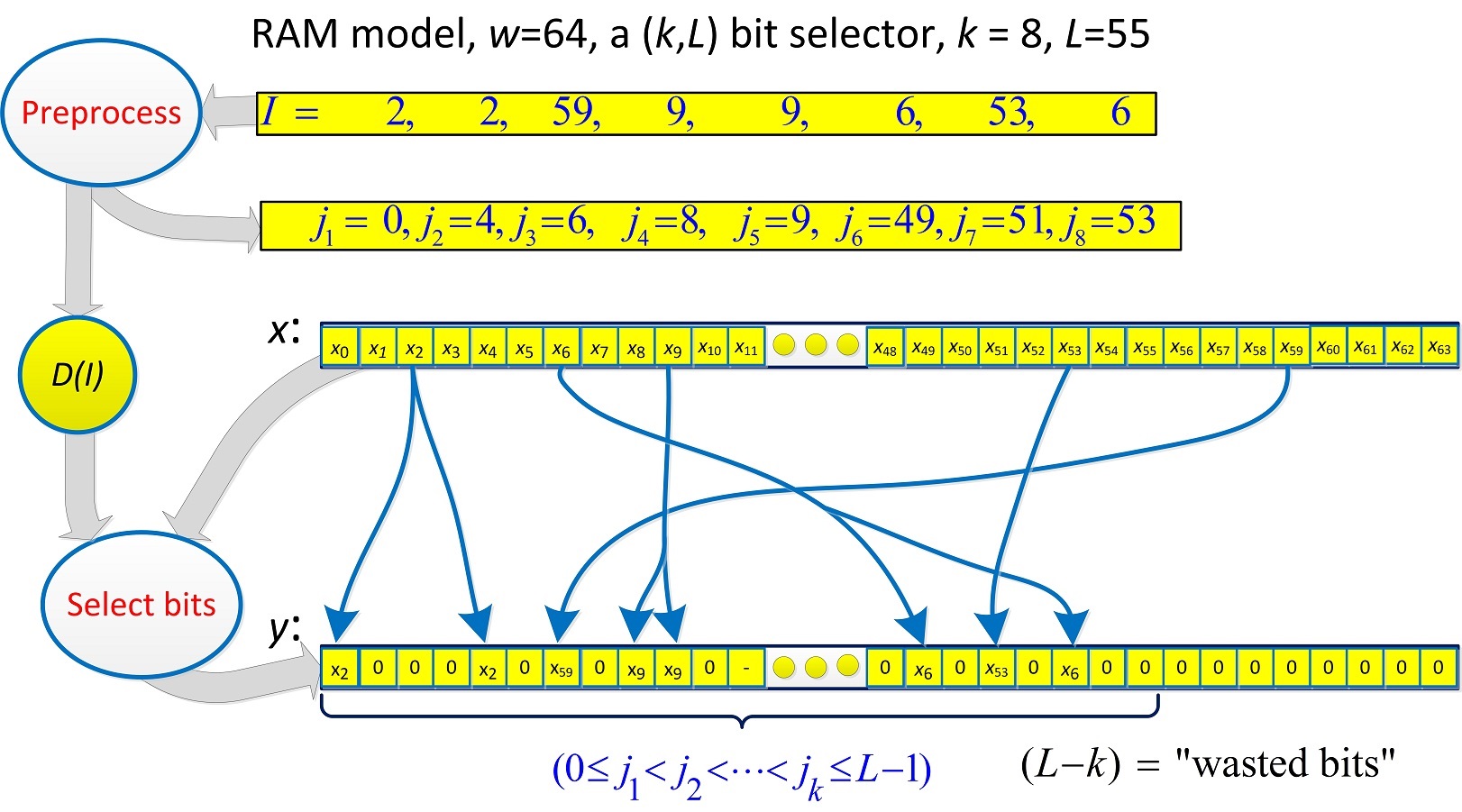}
\caption{A $(k,L)$ bit-selector. Note that $I[\ell]$, $1 \leq \ell
\leq k$, values can include repetitions and need not be in any
order, however the $j_\ell$ sequence is an ascending subsequence of
$1, 2, \ldots, L$.  } \label{fig-extractor1}
\end{figure*}


\section{$\beta$-nodes}
\label{sec:beta-full} We give an alternative single word index, the
$\beta$-node, which, like the $\alpha$-node, is randomized. It's
expected query time is $O(\log w)$. Its worst-case probe complexity
is inferior compared to the $\gamma$-node, but it may be simpler to
understand / implement than the $\alpha$-node. The $\beta$-node does
not require the  use of our bit-selectors, instead, like the
$z$-fast trie, it compares hash values.

Let $S_0$ be a set of $n$ $w$-bit binary strings stored
consecutively in ascending order in the memory. A $\beta$-structure
is a randomized succinct index data-structure, which supports
$rank_{S_0}(x)$ queries for any $x \in \{0,1\}^w$ (recall that
$rank_S(x)$ is the rank of $x$ in $S$). Its size is $O(n(\log w +
\log n) \cdot \frac{1}{w})$ $w$-bit words, the query time is $O(\log
n)$ (in expectation and w.h.p), and the number of probes it makes
outside the index is $O(1)$ (in expectation and w.h.p). Here w.h.p
means that the probability that a query will take $O(\log n)$ time
and $O(1)$ probes outside the index is at least $1-\frac{1}{n}$. By
a $\beta$-node we refer to a $\beta$-structure with $n =
\frac{w}{\log w}$ keys. The size of a $\beta$-node is $O(1)$, the
query time is $O(\log w)$ (in expectation and w.h.p) and the number
of probes outside the index is $O(1)$ (in expectation and w.h.p).

We start by describing a non-succinct version of a $\beta$-structure, which we refer to as $B$-structure, and
then we describe how to transform the non-succinct $B$-structure into a succinct $\beta$-structure
which occupies only $O(n(\log w + \log n))$ bits.

\subsection{The Prefix Partitioning Lemma} Let us start by defining
a prefix-partition operator $\sqsubseteq$: Let $S \in \{0,1\}^*$ be
an arbitrary set of binary strings, and let $p \in \{0,1\}^*$ be a
binary string, we partition $S$ into two subsets: $S_{\sqsubseteq
p}$ and $S_{\not\sqsubseteq p}$.
\newline $S_{\sqsubseteq p}$ is the set of all the elements of $S$ which start with $p$,
and $S_{\not\sqsubseteq p} = S - S_{\sqsubseteq p}$, is the set of all the elements of $S$ which
don't start with $p$.

The following lemma proves that there exists a prefix $p$ which partitions $S$ into
two approximately equal subsets.

\begin{theorem} \label{lemma:StringsSetPartition}
For every set $S$ of binary strings, $|S| \ge 2$, there exists a binary string $p \in \{0,1\}^*$ {\sl s.t.}
$\frac{1}{3}|S| \le |S_{\sqsubseteq p}| \le \frac{2}{3}|S|$  and
$\frac{1}{3}|S| \le |S_{\not\sqsubseteq p }| \le \frac{2}{3}|S|$
\end{theorem}
\begin{proof}
Let initially $p = \epsilon$ be the empty string. While $(|S_{\sqsubseteq p}| > \frac{2}{3}|S|)$,
if  $|S_{\sqsubseteq p0}| \ge |S_{\sqsubseteq p1}|$ add a $0$-bit to the end of $p$, otherwise add
a $1$-bit to the end of $p$.
We stop the loop at the first time that $|S_{\sqsubseteq p}| \le \frac{2}{3}|S|$, and it is easy to verify
that at that point we also have
that $|S_{\not\sqsubseteq p}| \le \frac{2}{3}|S|$.
\end{proof}

\subsection{Construction of a $B$-structure}
Let $S_0$ be the initial set of $n$ $w$-bit binary strings. Assume without loss of generality that $0^w \in S_0$.
We can assume that, since if $0^w \not\in S_0$ then $rank_{S_0}(x) = rank_{S_0 \cup \{0\}}(x) - 1$ for
every $x>0$.
The $B$-structure is a binary tree containing a prefix of the strings in $S_0$ in each internal node, and
at most $3$ keys of $S_0$ at each leaf.

We define the $B$-node recursively for a subset $S \subseteq S_0$ (starting with $S = S_0$).
Let $p$  be a prefix of a key in $S$ as in Lemma \ref{lemma:StringsSetPartition},
define $p_{min}, p_{max} \in S$ to be the
minimum and maximum keys respectively in $S$ which starts with $p$. Store $p$ in the root of the $B$-structure
of $S$, the right child is a $B$-structure of
$S_R = S_{\sqsubseteq p} \cup \mathrm{StrictPredecessor}(S,p_{min})$
(we define here $\mathrm{StrictPredecessor}(S,p_{min})$ to the the maximal key in $S$ which is
smaller than $p_{min}$, or $p_{min}$ itself if it is the minimal key in $S$), the left child is a
$B$-structure of $S_L = S_{\not\sqsubseteq p} \cup p_{max}$.
We "associate" $S$ with the root, $S_R$ with its right child,
and $S_L$ with its left child. When $|S| \le 3$, we stop the recursion and store the (at most 3) keys of $S$
in the leaf of the $B$-structure.

According to Lemma \ref{lemma:StringsSetPartition}, $|S_R|, |S_L| \le \frac{2}{3} |S| + 1$, and since
$|S_L| + |S_R| \le |S| +2$, we get that the height
of the resulting tree is $O(\log |S_0|)$, and the number of nodes in the tree is $O(|S_0|)$.

\subsection{Querying the $B$-node}
Let $x \in \{0,1\}^w$ be  the query word. Start the search in the root.
The root contains a prefix $p$, if $x$ starts with $p$ continue the search in the right child, otherwise, continue
the search in the left child. Continue the search similarly in every internal node that we reach,
until we reach a leaf.

In the leaf at most $3$ keys are stored, denote them by $s_1 < s_2 < s_3$. If $s_1 \le x < s_2$ output
$\mathrm{rank}_{S_0}(s_1)$ (that is, the rank of $s_1$ in $S_0$). If $s_2 \le x < s_3$ output
$\mathrm{rank}_{S_0}(s_2)$. Otherwise, $x > s_3$ output
$\mathrm{rank}_{S_0}(s_3)$.

\subsection{Correctness of the $B$-structure}
Given $x \in \{0,1\}^w$ we need to prove that the output of the search procedure is $\mathrm{rank}_{S_0}(x)$.

Let $y = \mathrm{Pred(S_0, x)}$
be the predecessor of $x$ in $S_0$.
At the end of the search we reach a leaf which stores between $1$ and $3$ elements of $S_0$.
We need to prove that $y$ is one of these keys. Let $(v_0, v_1, \ldots, v_k)$ be the root-to-leaf path
traversed during the search. Let $p_i$ be the prefix stored at $v_i$ and let $S_i$ be the subset of $S_0$
associated with $v_i$, for $i=0,1,\ldots,k$. We prove by induction on $i$, for $i = 0, 1, \ldots, k$,
that $y$ is a member of $S_i$. This is correct at the root (i=0), since $y \in S_0$ by our assumption
that $0^w \in S_0$.
For the inductive step, we assume $y \in S_i$, and prove that $y \in S_{i+1}$.

\begin{lemma}
Let $y = \mathrm{Pred}(S_0, x)$. If $y \in S_i$ then $y \in S_{i+1}$.
\end{lemma}
\begin{proof}
Let $p_{min}, p_{max} \in S_i$ be the minimum and maximum keys in $S_i$ respectively which start with $p_i$.

If $v_{i+1}$ is a right child,
then $x$ starts with $p_i$ and $S_{i+1} = (S_i)_{\sqsubseteq p_i} \cup \mathrm{StrictPredecessor}(S_i, p_{min})$.
If $y \in (S_i)_{\sqsubseteq p_i}$ we are done, otherwise it must be that $y = \mathrm{StrictPredecessor}(S_i, p_{min})$
since $x$ starts with $p_i$ and $y$ is its predecessor in $S_i$.

If $v_{i+1}$ is a left child,
then $x$ doesn't start with $p_i$ and $S_{i+1} = (S_i)_{\not\sqsubseteq p_i} \cup p_{max}$.
If $x < p_{min}$ then $y$ doesn't start with $p_i$ and hence $y \in (S_i)_{\not\sqsubseteq p_i}$.
If $x > p_{max}$ then either $y = p_{max}$, or $y > p_{max}$ and then $y \in (S_i)_{\not\sqsubseteq p_i}$.
We get  that in all the cases, $y \in (S_i)_{\not\sqsubseteq p_i} \cup p_{max}$ as required.
\end{proof}

\subsection{Making it succinct: from $B$-structure to $\beta$-structure}
We now describe the succinct variant of the $B$-structure, which we call $\beta$-structure.
Its index occupies
$O(n(\log w + \log n))$ bits, and its search time is $O(\log n)$ (w.h.p), and the number of probes outside
the index is $O(1)$ (w.h.p).

Each node of the $\beta$-structure occupies $\log w + \log n$ bits, defined as follows:
\begin{itemize}
\item \textbf{Inner nodes:} Replace every prefix $p$ in the $B$-structure with
a pair $\mathrm{<}|p|, h(p)\mathrm{>}$, $|p|$ is the length of the prefix $p$ ($\log w$ bits)
and $h(p)$ being a signature of $p$ of length $(2\log n)$ bits, computed using a universal
hash function. To test if $x$ starts with $p$ check if $h(x[1..|p|]) = h(p)$. If so, then $x$ starts with $p$
with probability at least $1-2^{-2\log n} = 1-(\frac{1}{n})^2$, otherwise $x$ doesn't start with $p$.

\item \textbf{Leaves:} In the leaves, replace the $O(1)$ (at most $3$) keys stored at each leaf with
their rank in $S_0$ ($O(\log n)$  bits).
In the search procedure, when reaching a leaf, retrieve these keys (from the static sorted list of the keys of $S_0$)
using their ranks in $O(1)$ word-accesses.
\end{itemize}

Finally, at the end of the search assume the algorithm suggests that $\mathrm{rank}_{S_0}(x) = i$. We can test
if it's correct using $O(1)$ word-accesses by checking that $s_i \le x < s_{i+1}$
(recall the notation $S_0 = s_1 < s_2 < \ldots < s_n$).
With probability at most $2^{-2\log n} \cdot \log n < \frac{1}{n}$ we will get an error at some node along the
search path. When an error is detected we
do a binary search to find the predecessor of $x$ among the static set of $n$ keys, this takes $O(\log n)$ time and
probes. So the binary search takes $o(1)$ on average.
When no error occurs, the search time is $O(\log w)$ (this happens with probability at least $1-\frac{1}{n}$).
Hence, the query time is $O(\log w)$ (in expectation and w.h.p), and we accesses only $O(1)$ words
outside the index (in expectation and w.h.p).

\section{Detailed Description of the \\ $(w/\log w, w/\log w)$ Bit-Selector}
\label{sec:extract-full}

In this section we describe both the preprocessing and selection
operations for our bit-selectors. We follow the selection process,
which makes use of $D(I)$, the output of the preprocessing. While
describing the selection process we specify the different components
of $D(I)$.

$D(I)$ consists of $O(1)$ words and includes precomputed constants
used during the selection process. As $D(I)$ is $O(1)$ words, we
assume that $D(I)$ is loaded into registers at the start of the
selection process. Also, the total working memory required
throughout the selection is $O(1)$ words, all of whom we assume to
reside within the set of registers.

Let $I= I[0], I[1], \ldots, I[k-1]$, $k \le w/\log w$, be the
sequence of bit indices to be selected from some input word $x$. $I$
may contain repetitions. The indices $I[j]$ range in $[0,w-1]$,
where bit zero is the most significant bit (on the left in the
Figures). Let $r$ be the number of distinct values in $I[0], \ldots,
I[k-1]$.

Partition the sequence $\sigma= 0,1, \ldots,w-1$ into $w/\log w $
blocks (consecutive, disjoint, subsequences of $\sigma$), each of
length $\log w$. Let $B_j$ denote the $j$th block of a word, {\sl
i.e.}, $B_j = j\log w, j\log w + 1, \ldots, (j+1)\log w -1 $,  $0
\le j \le w/\log w -1$.

We define the following notation. For sequences $\sigma$ and $\tau$,
$\sigma\cap \tau$ denotes a subsequence of $\sigma$ consisting of
those values that appear somewhere in $\tau$.
 We recall that for a sequence of
indices $\sigma$, we define $x[\sigma]$ to be the bits of $x$ in
these positions (ordered as in $\sigma$), if $\sigma$ has
multiplicities then $x[\sigma]$ also has multiplicities.

An assignment to $x[\sigma]$, such as $x[\sigma] \leftarrow b_1,
b_2, \ldots, b_{|\sigma|}$, $b_i\in \{0,1\}$, is shorthand notation
for $x[\sigma[1]] \leftarrow b_1$, $x[\sigma[2]] \leftarrow b_2$,
$\ldots$, $x[\sigma[|\sigma|]] \leftarrow b_{|\sigma|}$. (Assignment
to $x[\sigma]$ makes sense if $\sigma$ has no multiplicities).

Also, given a word $z$, let $z \gg i$ denote a right shift of $z$ by
$i$ bits, and $z \ll i$ a left shift by $i$ bits.

Given an input word $x$ and the precomputed $D(I)$, the selection
process goes through the seven phases described below.

\def\suff{{\rm suff}}
\subsection{The ever changing  $x$ and $I$}

As we process the various phases and sub phases of the bit
selection, the original bits of $x$ are permuted, duplicated, or set
to zero.

Phases $1 - 5$ and 7 simply permute the bits of $x$. Each
permutation is performed in $O(\log w)$ operations. Phase 6
duplicates some of the bits in $x$ (those bits with multiplicity $
> 1$ in the sequence $I$). Each of the phases requires
$O(1)$ precomputed words throughout its execution.

Let $x_0$ be the original word and $I_0=I$ be the original sequence
of indices. Moreover, let $x_t$ be the word $x$ after phases $1$ to
$t$, and let $I_t$ be a sequence of indices such that for all
$j=0,\ldots,k-1$, $x_t[I_t[j]]=x_0[I_0[j]]$.

For any $t$, $0<t\neq 6$, imagine that $I_t$ is obtained by changing
$I_{t-1}$ so as to reflect the bit permutation performed on
$x_{t-1}$ to get $x_t$. This permutation on $I_{t-1}$ need not be
actually done, these permutations are implicitly used by the bit
selection algorithm.

During phase 6, where bits are duplicated so that the number of
copies of each bit is equal to the multiplicity of the index of the
bit in $I_0$ (or $I_5$), imagine that $I_6$ is produced from $I_5$
by removing multiplicities and substituting $i+j-1$ for  the $j$th
appearance of index $i$ in $I_5$

It follows that for all $0 \leq t \leq 5$, $0 \leq j \leq k-1$, the
multiplicity of $I_t[j]$ is equal to the multiplicity of $I_0[j]$.
For $t=6,7$, $0 \leq j \leq k-1$, the multiplicity of $I_t[j]$ is
one.

Initially, all bits not appearing in $I_0$ are set to zero simply by
setting $x_0 \leftarrow x\ \AND\ M$ where $M$ is a mask with it's
$i$th bit equal to one iff $i$ appears in $I$.

The final output of the bit selection, from left to rights, is a
word $x_0[I] \| 0^{w-|I|}$.

For brevity, we use $x$ as a continuously changing variable
throughout the description of the different phases. The sequences
$I_t$ are needed during the preprocessing phase, the query phase
requires only a constant number of precomputed words. We describe
how the preprocessing phase keeps track of the various $I_t$
sequences and the permutations applied to $x$ implicitly through the
description of the phases.

\subsection{Phase 1: Packing blocks to the left (Figure \ref{fig:phase1overview-full})}

We now describe the procedure for rearranging the bits of $x$ so
that for all blocks $B$, the bits $x[B\cap I]$ are assigned to the
leftmost positions of $x[B]$, preserving their  order.  This will be
done for all blocks in parallel by the inherent parallelism of word
operations.

  For block $B$,
 let $\suff_i(B)$ be the length $i$ suffix of $B$.

Phase 1 requires $\log{w}$ subphases. We maintain the following
invariant after subphase $i$, $1\le i \le \log w$: The bits
 $x[\suff_i(B)\cap I]$ are assigned to the leftmost positions of
$x[\suff_i(B)]$ and the other bits of $x[\suff_i(B)]$  are set to
zero. Bits of $x$ whose indices are not in $\suff_i(B)$ do not
change. Note that this invariant initially holds for $i=1$.

 At subphase $i$, for $i=2,\ldots ,\log w$,
 for each block $B$ whose $i$th largest index is not in $I_0$
we assign $x[\suff_{i-1}(B)]\| '0'$ to $x[\suff_{i}(B)]$.

\begin{figure*}[htbp]
\centering
\includegraphics[scale=0.6]{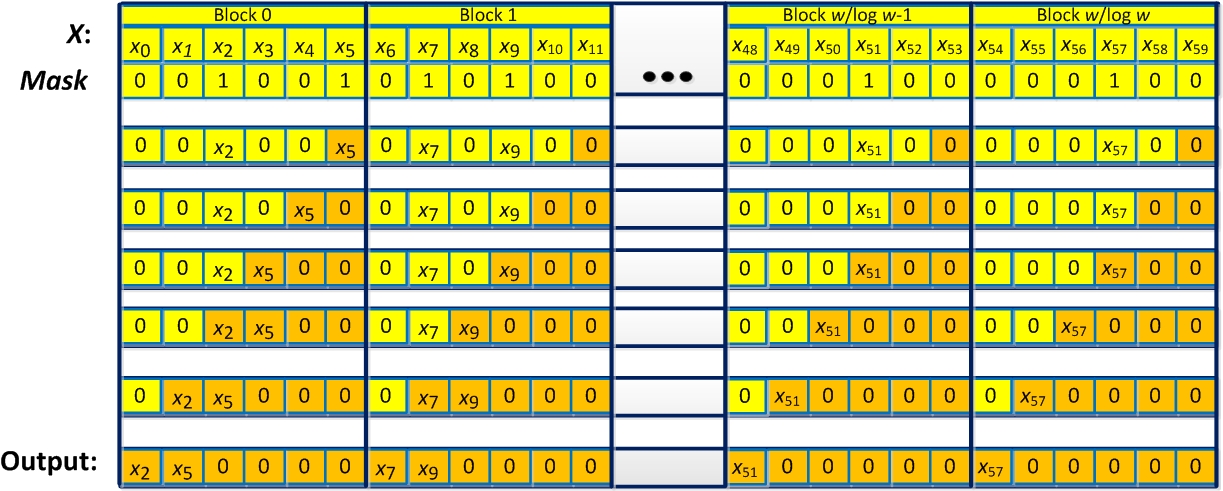}
\caption{An illustration of Phase $1$.
\label{fig:phase1overview-full}}
\end{figure*}

Let $Z_i$ be a word with $1$ at the $i$th largest index of each
block, and zeros elsewhere. We need $Z_i$ during subphase $i$ of
Phase 1. $Z_1$ can be constructed on the fly in a register, in time
$O(\log w)$, $Z_i$ is simply a left shift of $Z_1$ by $i-1$. Let
$L_i = \overline{M}\ \AND\ Z_i$, and let
 $S_i
= L_i - (L_i \gg (i-1))$. See Figure \ref{fig:phase1masks-full}.

\begin{figure*}[htbp]
\centering
\includegraphics[scale=0.5]{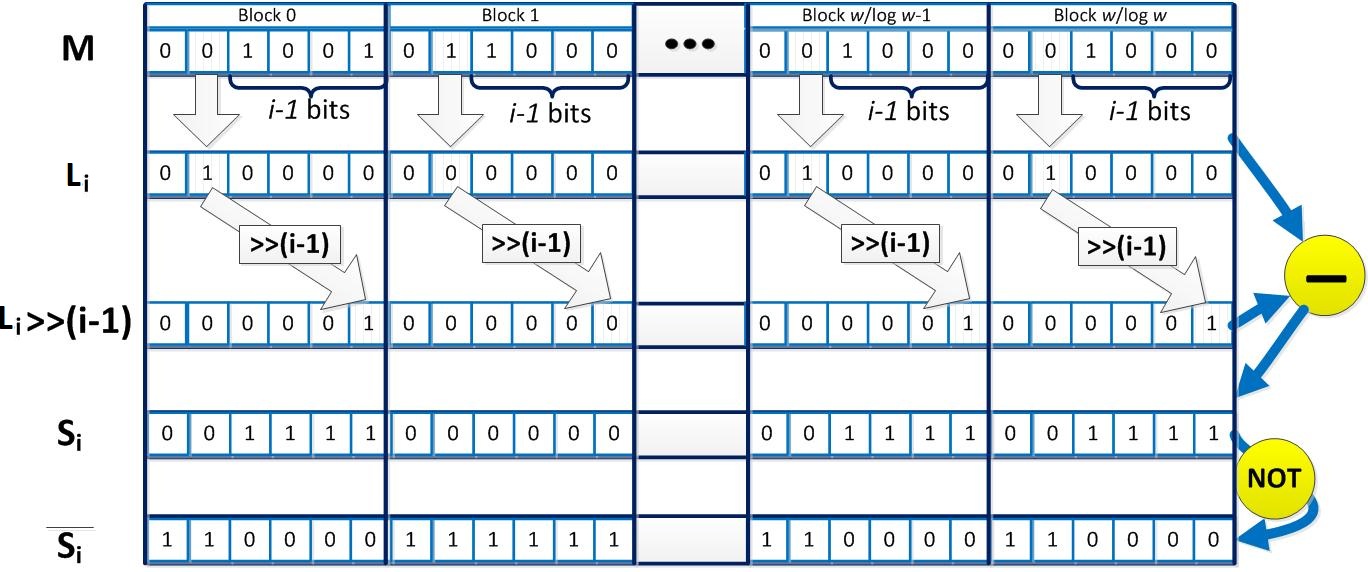}
\caption{The masks used during Phase 1.
\label{fig:phase1masks-full}}
\end{figure*}

\begin{figure*}[htbp]
\centering
\includegraphics[scale=0.5]{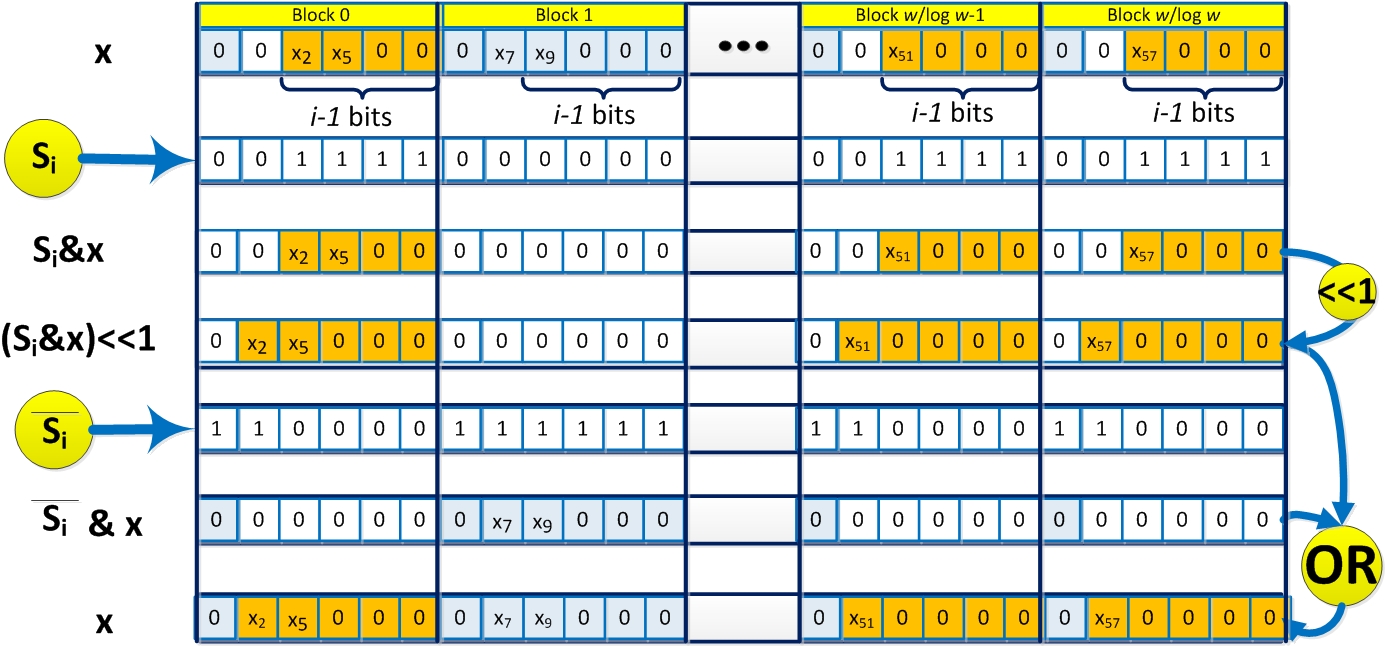}
\caption{A subphase of Phase 1. \label{fig:phase1iterations-full}}
\end{figure*}

The $i$th subphase is as follows: We compute $y_1=x\ \AND\ S_i$
which gives the bits that have to be left shift by one position, and
we compute $y_2=x\ \AND\ \bar{S_i}$ which gives a word containing
the bits which are to remain in their positions. Finally, we set $x
= (y_1 \ll 1)\ \OR\ y_2$. See Figure
\ref{fig:phase1iterations-full}.

\subsection{Phase 2: Sorting blocks by size (Figure \ref{fig:phase2-full})}

We permute $x[B_0],\ldots,x[B_{\frac{w}{\log w}-1}]$ such that they
are in non-increasing order of $|B_j\cap I_1|$.  Note that we know
this permutation when preprocessing $I_1$. We implement this step
using a simulation of a Benes-network (described in Section
\ref{subsec:benes-full}). This simulation requires
 $O(\log w)$ operations, and uses
$O(1)$ precomputed constants stored in $D(I)$, and $O(1)$ registers.

\begin{figure*}[htbp]
\centering
\includegraphics[scale=0.5]{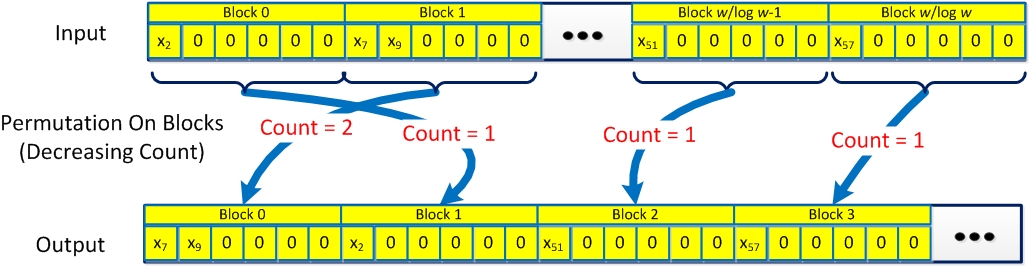}
\caption{An illustration of Phase 2. \label{fig:phase2-full}}
\end{figure*}

\subsection{Phase 3: Dispersing bits (Figure \ref{fig:phase3-full})}

Recall that $r$ is the number of distinct values in $I$ . Let
$\sigma = B_0[0] , B_1[0] , \ldots , B_{r-1}[0]$, {\sl i.e.},
$\sigma$ is a sequence of the first (and smallest) index in every
block. For any sequence $\pi$, define $\tau(\pi) = \tau_1, \tau_2,
\ldots, \tau_r$ be a subsequence of $\pi$ where $\pi[j]$ is
discarded if $\pi[j']=\pi[j]$ for some $j'<j$.

 In Phase 3 we  disperse the bits  of $x[I_2]$, so that
$$x[\sigma] \leftarrow x[\tilde{\tau}],$$ for some sequence
$\tilde{\tau}$ produced by some permutation on the order of
$\tau(I)$. The description of $\tilde{\tau}$, is implicit in the
description of Phase 3 below.

Following Phase 2, we have that  $|B_0 \cap I_2| \ge |B_1 \cap I_2|
\ge \ldots \ge |B_{w/\log w-1}\cap I_2|$. Therefore, for $1\leq i
\leq \log w$, we can define $$a_i = |\{j |  |B_j\cap I_2|\geq i
\}|.$$ Let $A_i = \sum_{\ell=1}^{i} a_\ell$ for $i=1,\ldots,\log w$,
and define $A_0=0$.

We can now define the sequence $\sigma_i$, $1\leq i \leq \log w$,
$$\sigma_i= \sigma[A_{i-1}\log w], \sigma[(A_{i-1}+1)\log w], \ldots, \sigma[(A_i-1)\log w],$$ and the sequence
$$\xi_i = B_0[i-1], B_1[i-1], \ldots, B_{a_{i}-1}[i-1], \qquad 1\leq i \leq \log w.$$

Phase 3 has  $\log{w}$ subphases. Subphase $i$ of Phase 3 performs
the assignment $x[\sigma_i] \leftarrow x[\xi_i]$, this assignment
can be implemented using $O(1)$ operations.

Isolate the bits to be moved (indices $\xi_i$), shift them to their
new locations (indices $\sigma_i$, note that $\sigma_i[j]-\xi[j] =
\sigma_i[j']-\xi[j']$ for all $1\leq j,j'\leq |\sigma_i|$),
producing word $y$. Next,update $x$ by setting $x[\xi_i]$ to zero
and taking the OR with $y$.

The values $A[i]\log w$ are stored as part of $D(I)$ (in total
$\log^2 w$ bits). These values suffice so as to generate all the
masks and operations required in Phase 3.

\begin{figure*}[htbp]
\centering
\includegraphics[scale=0.5]{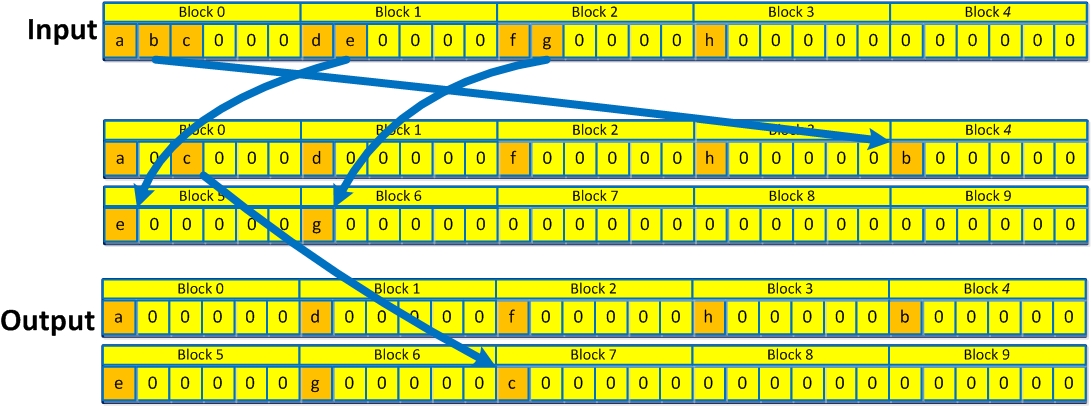}
\caption{An illustration of Phase 3. \label{fig:phase3-full}}
\end{figure*}

\subsection{Phase 4: Packing bits (Figure \ref{fig:phase4-full})}

 Let $\sigma$ and $\tau(\pi)$  be as defined at the
start of Phase 3. in Phase 4 we ``push" the bits $x[\tau(I_3)]$ to
the left, {\sl i.e.},
$$x[0,1, \ldots, r-1] \leftarrow x[\tilde{\tau}],$$
where $\tilde{\tau}$ is a sequence produced by some permutation on
the order of $\tau(I_3)$.  As in Phase 3, the description of
$\tilde{\tau}$, is implicit in the description of Phase 4 below.
Note that $\tau(I_3)$ is a permutation of $\sigma$.

There are $\log w$ subphases in Phase 4.

Let $q=\lfloor r/\log w \rfloor$. In subphases $1,\ldots,\log w-1$
we fill $x[B_0]$, $x[B_1]$, $\ldots$, $x[B_{q-1}]$, with some
permutation of the first $q \log w$ bits of $x[\sigma]$. The last
subphase is used to copy the $k< \log w$ leftover bits of
$x[\sigma]$ into $x[B_q[0], B_q[1], \ldots, B_q[k-1]]$.

For $1 \leq i < \log w$ define the sequences $\upsilon_i = i, \log w
+ i, \ldots, (q-1)\log w + i$ and $\zeta_i = i q \log w, (iq + 1)
\log w, \ldots, ((i+1)q -1) \log w $.

 In Subphase $i$ of Phase
4 we perform the assignment $$x[\upsilon_i] \leftarrow x[\zeta_i].$$

To do this using word operations, we first isolate the bits of
$x[\zeta_i]$, shift them so as to be in their target locations,
$\upsilon_i$, and ``or" them into place.

The last subphase copies the remaining bits one by one, for a total
of $O(\log w)$ operations.

\begin{figure*}[htbp]
\centering
\includegraphics[scale=0.5]{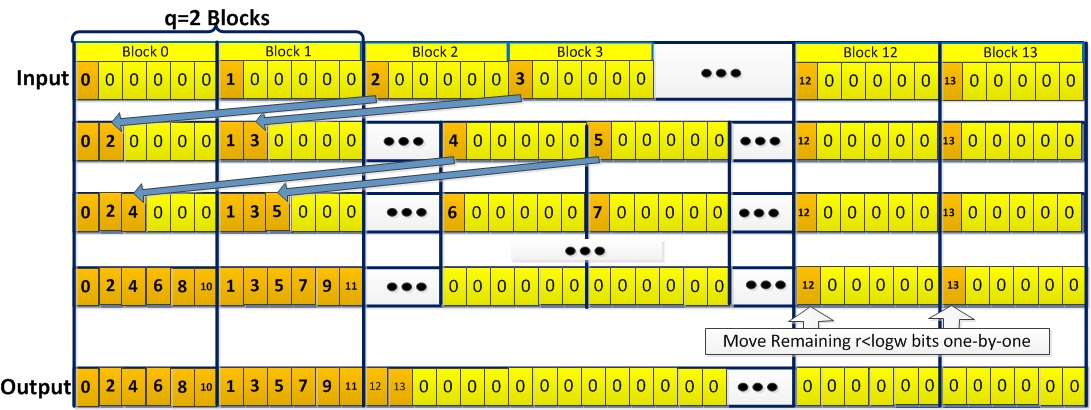}
\caption{An illustration of Phase 4. \label{fig:phase4-full}}
\end{figure*}

\subsection{Phase 5: Spacing the bits (Figure \ref{fig:phase5-full})}
At the end of Phase 4 $x[0,1, \ldots, r-1]$ is a permutation of the
bits of $x[\tau[I_4]]$, and $x[j]$, $j\geq r$, are zero. Our goal is
now to space the bits so as to make space for duplication of those
bits whose indices appear multiple times in $I_4$.

In this phase we space the bits $x[0,1, \ldots, r-1]$ by
``inserting" $j-1$ zeros between $x[\ell]$  and $x[\ell+1]$ iff
$\ell$ appears $j$ times in $I_4$. We do this by permuting the bits
of $x[0,1, \ldots, k]$. There is one unique permutation that
achieves this goal. This is done by simulating a Benes sorting
network, in time $O(\log w)$, and using only $O(1)$ precomputed
constants and $O(1)$ registers.

\begin{figure*}[htbp]
\centering
\includegraphics[scale=0.5]{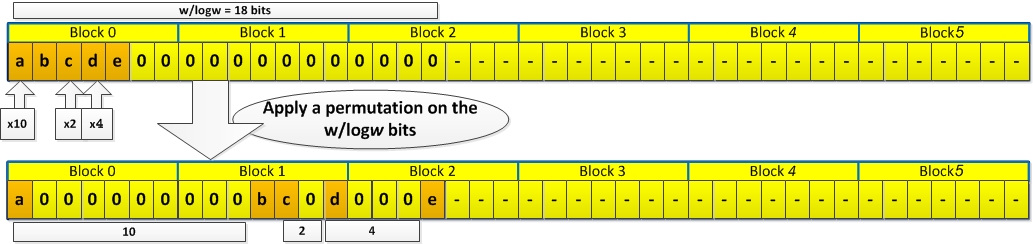}
\caption{An illustration of Phase 5. \label{fig:phase5-full}}
\end{figure*}

\subsection{Phase 6: Duplicating bits (Figure \ref{fig:phase6-full})}

For a sequence $\varrho$ let $m_\ell(\varrho)$ be the number of
occurrences of $\ell$ in $\varrho$. At the end of Phase 5, for every
$\ell\in I_5$ such that $m_\ell(I_5) > 1$  we have that $x[\ell+1,
\ell+2, \ldots, \ell+m_\ell(I_5)-1]$ contain zeros and none of the
indices $\ell+1, \ell+2, \ldots, \ell+m_\ell(I_5)-1$ appear in
$I_5$.

Phase 6 consists of $\log w$ subphases, $i=1, \ldots, \log w$.
Subphase $i$ duplicates a subset of the bits of $x$ specified by a
$w/\log w$ bit mask $M_i$. All these masks a precomputed at
preprocessing and store in a single word with $D(I)$.

We compute the masks $M_i$ as follows. Let $I^i_6$ be the sequence
which describes the positions of the bits of $I$ at the end of
 subphase $i$, and let $I^0_6=I_5$. When a bit is copied
 we  split its remaining multiplicity among the two copies.

Let $\Delta_i = 2^{\log w -i}$.  Subphase $i=1,\ldots,\log w$
duplicates those bits $x[\ell]$ for which
\begin{equation}m_\ell(I^{i-1}_6) > \Delta_i.\label{eq:mell-full}\end{equation}
So $M_i$ is set to one in all positions $\ell$ for which Equation
(\ref{eq:mell-full}) holds and it is zero in all other places.

 $I^i_6$ is computed from $I^{i-1}_6$ as follows: For every $\ell$
that appears somewhere in $I^{i-1}_6$ let $i_1<i_2< \cdots < i_t$ be
the indices of all occurrences of $\ell$ in $I^{i-1}_6$. Let
$I^i_6[i_j] = \ell$ for $j<\Delta_i$ (unchanged from
$I^{i-1}_6[i_j]$), and set $I^i_6[i_j] = \ell + \Delta_i$ otherwise.
This effectively splits the multiplicity of $\ell$ between $\ell$
and its new copy $\ell + \Delta_i$. Bit $\ell$ has now multiplicity
$\Delta_i-1$ and bit $\ell + \Delta_i$ has the remaining
multiplicity.

At query time in subphase $i$ we set $x=(x\ \OR\ ((x\ \AND\ M_i) \gg
\Delta_i))$.

%
%
%

\begin{figure*}[htbp]
\centering
\includegraphics[scale=0.5]{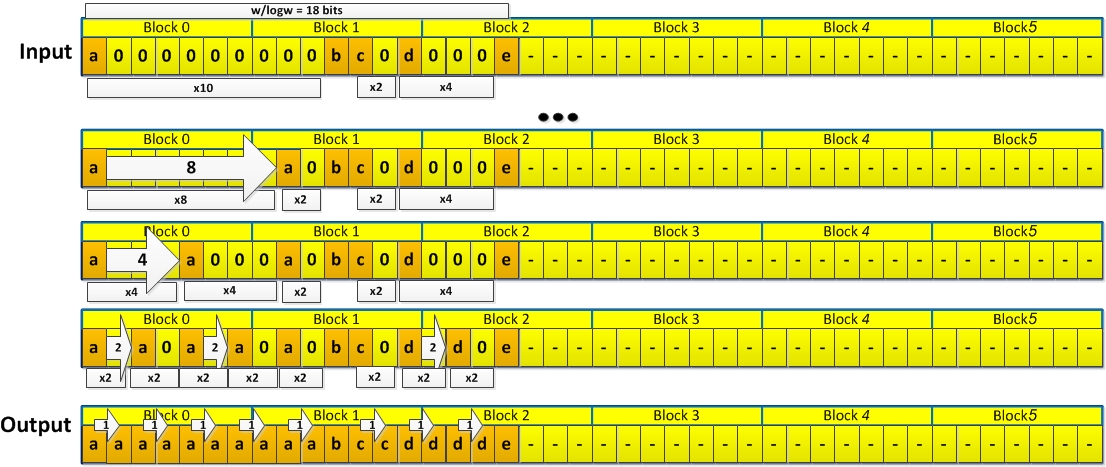}
\caption{An illustration of Phase 6. \label{fig:phase6-full}}
\end{figure*}

\subsection{Phase 7: Final Positioning (Figure \ref{fig:phase7-full})} At this stage we need permute the
bits $x[1, \ldots, k]$, so as to get the final output. Note that
$I_6$ is a permutation and it's inverse permutation, $I_6^{-1}$ is
the permutation we need to apply to $x$. This too requires
simulation of a Benes network, see Section \ref{subsec:benes-full}.

\begin{figure*}[htbp]
\centering
\includegraphics[scale=0.5]{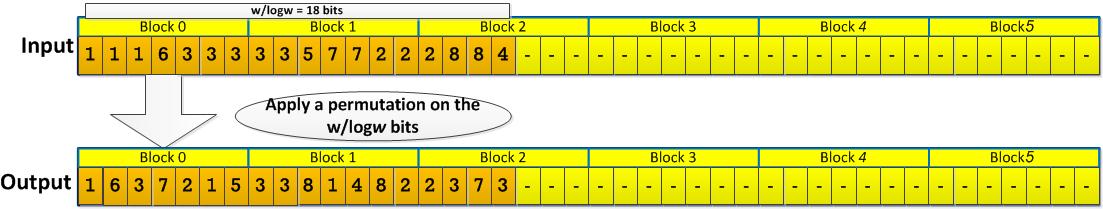}
\caption{An illustration of Phase 7. \label{fig:phase7-full}}
\end{figure*}

\subsection{Permuting elements in a word
by simulating a Benes network} \label{subsec:benes-full}

We show how to prepare a set $C$ of $O(1)$ words such that given $C$
a Benes network implementing a given permutation  $\sigma$ can be
applied to a word $x$ in $O(\log w)$ operations (shift, and, or).

We use such networks in two contexts:
\begin{itemize}
    \item To permute the  $b\leq w/\log w$ leftmost bits of
    $x$. We need this  in Phases 5 and  7 of bit selection.
    \item To permute $w/\log w$ blocks of
    bits (each block of length $\log w$). We need this during Phase 2
    of bit selection.
\end{itemize}

\begin{figure*}[htbp]
\centering
\includegraphics[scale=0.6]{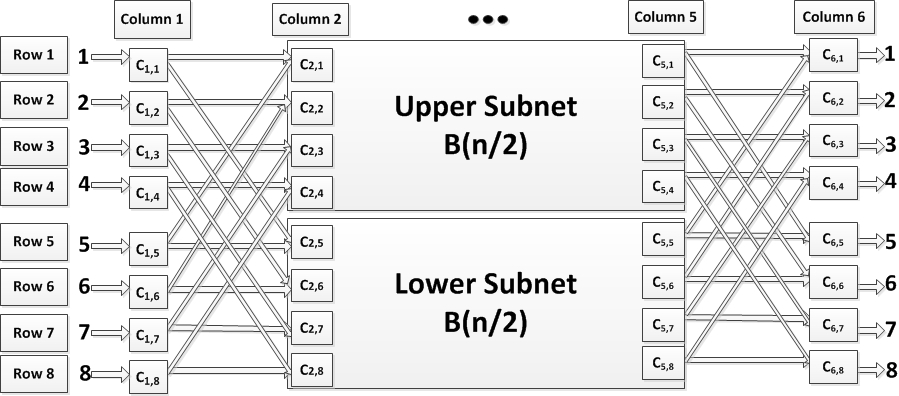}
\caption{One of many variants for the Benes network.
\label{fig:benes-full}}
\end{figure*}

\subsubsection{Overview of the Benes-Network}
\def\uu{\mbox{\rm u}}
\def\d{\mbox{\rm d}}
\def\Dir{\mbox{\rm Dir}}

Assume that $n$ is a power of $2$. A Benes network, $B(n)$, of size
$n$ consists of two Benes networks of size $n/2$, $B_{\uu}(n/2)$,
and $B_{{\d}}(n/2)$. For $1\le i\le n/2$, inputs $i$ and  $i+n/2$ of
$B(n)$ can be routed  to the $i$th input of  $B_{\uu}(n/2)$ or to
the $i$th input of $B_{\d}(n/2)$. The outputs are connected
similarly. For every $1\le i\le n/2$ we define inputs $i$ and
$i+n/2$ as {\em mates}, analogously we define outputs $i$ and
$i+n/2$ as {\em mates}. Note that mates cannot both be routed to the
same subnetwork. See Figure \ref{fig:benes-full}.

\noindent{\bf The looping algorithm}: A Benes network can realize
any permutation $\sigma$ of its inputs as follows. We start with
 an arbitrary input, say $1$, and route it to $B_{\uu}(n/2)$. This
 implies that the output $\sigma(1)$ is also routed to
 $B_{\uu}(n/2)$. The mate $o$ of $\sigma(1)$ must then be routed to
 $B_{\d}(n/2)$. This implies that $\sigma^{-1}(o)$ is routed to
$B_{\d}(n/2)$. If the mate of $\sigma^{-1}(o)$ is $1$ we ``completed
a cycle'' and we start again with an arbitrary input which we
haven't routed yet. Otherwise, if the mate of $\sigma^{-1}(o)$ is
not $1$ then we route this mate to $B_{\uu}(n/2)$ and repeat the
process.

\noindent{\bf Levels of the Benes network}: If we lay out the Benes
network then the 1st level of the recursion above gives us 2
``stages" consisting of $n$ $2 \times 2$ switches, stage $\#1$
connected to the inputs to $B_{\uu}(n/2)$ and $B_{\d}(n/2)$, and
stage $\#$ $2\log n -1$ connecting $B_{\uu}(n/2)$ and $B_{\d}(n/2)$
to the outputs. Opening the recursion gives us $2 \log n -1$ stages,
each consisting of $n$ $2\times 2$ switches.

To implement any specific permutation, one needs to set each of
these switches.

\subsubsection{Permuting the $b=w/\log w$ leftmost bits of the word}
\label{subsubsec:bitperm-full}
 We now  describe an $O(1)$ word representation for
 any permutation $\sigma$ on  $b=w/\log w$
elements that allows us to apply $\sigma$ to the $b$ leftmost bits
of a query word $x$ while doing only $O(\log b)$ operations. We
obtain this data structure by encoding the Benes network for
$\sigma$ in $O(1)$ words. To answer a query we use this encoding to
apply each of the $2 \log (b) -1 = O(\log b)$ stages of the Benes
network for $\sigma$ to the leftmost $b$ bits of $x$. Every stage
requires $O(1)$ operations giving a total of $O(\log b)=O(\log w)$
operations.


During preprocessing we prepare two $b \times (2 \log b -1)$ binary
matrices $\Dir$ and $C$. Both $\Dir$ and $C$ have $2 b \log b -b
\leq 2 w$ bits, so they can fit into 2 $w$-bit words. The $j$th
column of these matrices correspond to stage $j$ of the Benes
network, the $i$th row of these matrices corresponds to the $i$th
input of the stage. Pictorially, we imagine that inputs are numbered
top-down.

Recall that the mate of input $i$ in stage $j$ is some other input
$i'$ of stage $j$. If $(i'<i)$ we define $\Dir_{i,j}=0$, otherwise,
$(i'>i)$, and we define $\Dir_{i,j}=1$, this is defined for
$i=1,\ldots,b, j= 1,\ldots, 2\log b-1$.

The matrix $C$ is computed as follows: $C_{i,j} =0$ if input $i$ of
stage $j$ routes to input $i$ of stage $j+1$ (i.e. goes
``straight"), and $C_{i,j} =1$ otherwise.

We pack the $b \times (2 \log b -1)$ binary matrix $C$ into 2
$w$-bit words $C_1$, $C_2$ as follows:
\begin{eqnarray*}
  C_1[1,2,\ldots, w] &\leftarrow& C_{1,1}, \ldots, C_{1,b}, C_{2,1},
  \\&& \ldots, C_{2,b},
 \ldots, C_{\log b,1}, \ldots, C_{\log b,b}; \\
C_2[1,2,\ldots, w] &\leftarrow& C_{\log b+1,1}, \ldots, C_{\log
b+1,b}, \\&& \ldots, C_{2\log b -1,1}, \ldots, C_{2\log b-1,b};
\end{eqnarray*}

We pack the matrix $\Dir$ into words $\Dir_1$ and $\Dir_2$
analogously.

\begin{figure*}[htbp]
\centering
\includegraphics[scale=0.55]{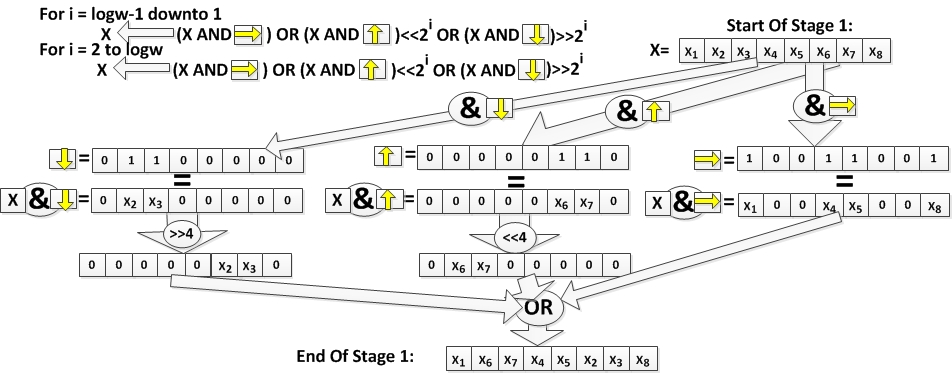}
\caption{Applying stage $1$ of the Benes network to $x$.
\label{fig:BenesNetworkStage-full}}
\end{figure*}

During query processing we apply stage $i$ (for $i=1 \ldots \log b$)
of the Benes network by computing
 \begin{eqnarray*} x &=& \left(x\ \AND\ \overline{C_1[1,\ldots,b]}\right)\\
 &\OR&  \left((x\ \AND\ C_1[1,\ldots,b]\ \AND\ \Dir_1[1,\ldots,b])\gg(b/2^i)\right)\
\\ &\OR& \left((x \ \AND\ C_1[1,\ldots,b] \ \AND\
\overline{\Dir_1[1,\ldots,b]})\ll(b/2^i)\right).\label{eq:benes-full}\end{eqnarray*}
This should be parsed as follows:
\begin{itemize}
  \item $(x \ \AND\ \overline{C_1[1,\ldots,b]})$ gives the bits of $x$ that are
  not going to change position at stage $i$.
  \item $(x \ \AND\ C_1[1,\ldots,b] \ \AND\ \Dir_1[1,\ldots,b])\gg(b/2^i)$ takes the bits of
  $x$ that are to move ``up" at stage $i$ and shifts them accordingly.
  \item $(x \ \AND\ C_1[1,\ldots,b] \ \AND\ \overline{\Dir_1[1,\ldots,b]})\ll(b/2^i))$ takes
  the bits of $x$ that are to move ``down" at stage $i$ and shifts
  them accordingly.
\end{itemize}

In preparation for the next stage we also compute $C_1 = C_1 \ll b$,
$\Dir_1 = \Dir_1 \ll b$ to prepare the control bits for the next
stage of the Benes network.

Analogously, during stages $i=\log b+1,\log b+2,\ldots,2\log b-1$ we
use the words $C_2$ and $\Dir_2$ rather than $C_1$ and $\Dir_1$.

An example of applying stage $1$ of a Benes network of size $8$ is
shown in figure \ref{fig:BenesNetworkStage-full}.

\subsubsection{Permuting the $w/\log w$ leftmost blocks of the word}

To operate the permutation on blocks of bits, we need masks that
replicate the appropriate $C_{i,j}$ and $\Dir_{i,j}$ values $\log w$
times so that they operate upon all bits of the block simultaneously
and not only on one single bit. To precompute and store such
replications in advance requires $O(\log w)$ words of storage, and
we allow, in total, only $O(1)$ words of storage for the entire bit
selection. Thus, we need compute these ``expansions" on the fly, and
in $O(\log w)$ operations.

We now add all-zero columns on the left of matrices $C$ and $\Dir$
so that each of them they have exactly $2 \log w$ columns. This is
well defined because $2 \log b -1 \leq 2 \log w$. Let these new
matrices be $\tilde{C}$ and $\widetilde{\Dir}$. Let
$\widetilde{C}^\mathrm{R}$ be the rightmost $\log w$ columns of
$\widetilde{C}$, and let $\widetilde{C}^\mathrm{L}$ be the leftmost
$\log w$ columns of $\widetilde{C}$. Also, let
$\widetilde{Dir}^\mathrm{R}$, and $\widetilde{\Dir}^\mathrm{L}$ be
defined analogously.

Previously, we packed the $C$ and $Dir$ matrices into words ($C_1$,
$C_2$) and ($\Dir_1$,$\Dir_2$), respectively, column by column. To
perform Block permutations we do so row by row as follows: The
matrix $\widetilde{C}^{\mathrm{L}}$ is packed into the word $C_1$,
row by row. Likewise, $\widetilde{C}^{\mathrm{R}}$ is packed, row by
row, into $C_2$, $\widetilde{\Dir}^{\mathrm{L}}$ into $\Dir_1$ and
$\widetilde{\Dir}^{\mathrm{R}}$ into $\Dir_2$.

The $C$ and $\Dir$ bits associated with stage $j$ of the Benes
network will be spaced out, $\log w$ bits apart. See Figure
\ref{fig:benesmasks3-full}. For $j< \log w$ these bits are in $C_1$
and $\Dir_1$.

\begin{figure*}[htbp]
\centering
\includegraphics[scale=0.6]{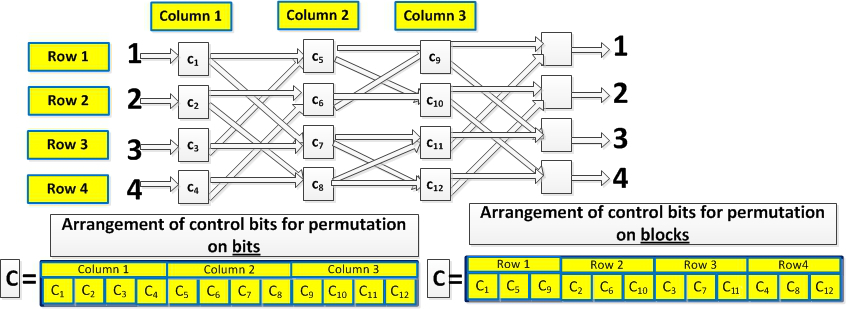}
\caption{Benes network control bits $C_{i,j}$ and $\Dir_{i,j}$.
\label{fig:benesmasks3-full}}
\end{figure*}

Given $C_i$ or $\Dir_i$, we seek to isolate and replicate the bits
associated with stage $1\leq j\leq \log w$. We define a
transformation $g: \{0,1\}^w \times \{1,\ldots, \log w\}$ such that
for any $w$ bit word $z$, and any $1 \leq j \leq \log w$, $g(z,j)$
is a mask such that for any block $B$, all bits of $g(z,j)[B]$ are
equal to $z[B[j]]$.

We compute $g(z,j)$ in $O(1)$ time as follows: Let $Z_j$, $1 \leq j
\leq \log w$, be a bit pattern with $w/\log w$ 1's at the $j$th
index of every block. The operation $y_0=Z_j\ \AND\ z$ isolates the
$j$'th bits of every block in $z$. Let $y_1 = y_0 \ll j-1$ and $y_2=
y_0 \gg (\log w- j)$, let $y_3=y_1-y_2$. Blocks $B$ for which the
bit $z[B[j]]=1$, now have $y_3[B] = 011\ldots1$, blocks $B$ for
which the  bit was zero now have $y_3[B]$ consisting only of zeros.
Finally, set $y_4=y_3\ \OR\ y_1$. Now, all bits of $y_4[B]$ are
equal to the bit $z[B[j]]$.

To simulate the Benes network and sort blocks rather than bits, for
stages $j\leq \log w$ we use the masks $g(C_1,j)$ and $g(\Dir_1,j)$,
analogously to our use of the masks $C_1[1,\ldots,b]$ and
$\Dir_1[1,\ldots,b]$ as used in Equation \ref{eq:benes-full}. For
stages $j>\log w$ we use $g(C_2,j-\log w)$ and $g(\Dir_2,j-\log w)$
analogously to the use of $C_2[1,\ldots,b]$ and
$\Dir_2[1,\ldots,b]$.

Given this transformation, we can simulate the Benes network in
parallel, on entire blocks, and permute blocks at no greater cost
than permuting $bits$.


\begin{thebibliography}{4}

\bibitem{Andersson:ac0}
A.~Andersson, P.~B. Miltersen, and M.~Thorup.
\newblock Fusion trees can be implemented with AC(0) instructions only.
\newblock {\em Theor. Comput. Sci.}, 215(1-2):337--344  1999.

\bibitem{Beame01optimalbounds}
P.~Beame and F.~E. Fich.
\newblock Optimal bounds for the predecessor problem and related problems.
\newblock {\em Journal of Computer and System Sciences}, 65(1):38--72  2002.

\bibitem{Belazzougui:2008}
D.~Belazzougui, P.~Boldi, R.~Pagh, and S.~Vigna.
\newblock Theory and practice of monotone minimal perfect hashing.
\newblock {\em J. Exp. Algorithmics}, 16:3.2, 2011.

\bibitem{Belazzougui:2009}
D.~Belazzougui, P.~Boldi, R.~Pagh, and S.~Vigna.
\newblock Monotone minimal perfect hashing: searching a sorted table with o(1)
  accesses.
\newblock SODA, 2009, 785--794.

\bibitem{Belazzougui10}
D.~Belazzougui, P.~Boldi, R.~Pagh, and S.~Vigna.
\newblock Fast prefix search in little space, with applications.
\newblock ESA,
  2010, 427--438.


\bibitem{DBLP:conf/wads/BrodnikMM97}
A.~Brodnik, P.~B. Miltersen, and J.~I. Munro.
\newblock Trans-dichotomous algorithms without multiplication - some upper and
  lower bounds.
\newblock WADS, 1997,  426--439.

\bibitem{Drepper07whatevery}
U.~Drepper.
\newblock What every programmer should know about memory, 2007,
http://lwn.net/Articles/250967/.


\bibitem{Ferragina:1999}
P.~Ferragina and R.~Grossi.
\newblock The string {B}-tree: a new data structure for string search in
  external memory and its applications.
\newblock {\em J. ACM}, 46:236--280, 1999.


\bibitem{FredmanWillard}
M.~L. Fredman and D.~E. Willard.
\newblock Surpassing the information theoretic bound with fusion trees.
\newblock {\em Journal of Computer and System Sciences}, 47(3):424 -- 436,
  1993.

\bibitem{FredmanWillard2}
M.~L. Fredman and D.~E. Willard.
\newblock Trans-dichotomous Algorithms for Minimum Spanning Trees and Shortest Paths.
\newblock FOCS, 1990, 719 -- 725.

\bibitem{Miltersen96}
P.~B. Miltersen.
\newblock Lower bounds for static dictionaries on rams with bit operations but
  no multiplication.
\newblock  ICALP, 1996, 442--453.



\bibitem{DBLP:conf/icalp/MunroRRR03}
J.~I. Munro, R.~Raman, V.~Raman, and S.~S. Rao.
\newblock Succinct representations of permutations.
\newblock ICALP, 2033, 345--356.


\bibitem{Patrascu:lower1}
M.~P\v{a}tra\c{s}cu and M.~Thorup.
\newblock Time-space trade-offs for predecessor search.
\newblock  STOC, 2006, 232--240.


%

\bibitem{Shavit}
N.~Shavit.
\newblock Data structures in the multicore age.
\newblock {\em Commun. ACM}, 54(3):76--84,  2011.

\bibitem{thorup03}
M.~Thorup.
\newblock On AC0 implementations of fusion trees and atomic heaps.
\newblock  SODA, 2003,  699--707.

\bibitem{DBLP:journals/ipl/Boas77}
P.~van Emde~Boas.
\newblock Preserving order in a forest in less than logarithmic time and linear
  space.
\newblock {\em Inf. Process. Lett.}, 6(3):80--82, 1977.


\bibitem{Dan:1983}
D.~E. Willard.
\newblock Log-logarithmic worst-case range queries are possible in space t(n).
\newblock {\em Information Processing Letters}, 17(2):81--84, 1983.


\end{thebibliography}
\end{document}